\documentclass[times,sort&compress,3p]{elsarticle} 
\journal{Journal of Multivariate Analysis}
\usepackage{amsmath, amssymb, enumerate, amsthm}
\usepackage{lineno}
	\usepackage{mathrsfs,anysize,fancyhdr,epsfig}
	\usepackage{graphicx}
    \usepackage{epstopdf}
	\usepackage[english]{babel}
	\usepackage[titletoc,title]{appendix}	
	\usepackage{subfiles}
    \usepackage{caption}
    \usepackage{subcaption}
    \usepackage{multirow}
    \usepackage{hyperref}
    \usepackage{geometry}
    \usepackage{rotating}
    \usepackage{tabularx}
    \usepackage[usenames,dvipsnames,svgnames,table]{xcolor}
    \usepackage{listings}
    \usepackage[framemethod=tikz]{mdframed}
    \usepackage[]{algorithm2e}
    \usepackage[nomarkers,figuresonly,nofiglist]{endfloat}
\usepackage{booktabs}
		\newcommand{\e}{\mathrm{e}}

		\newcommand{\diffp}[1]{ \frac{\partial}{\partial\,#1} }

        \newtheorem{theorem}{Theorem}
        \newtheorem{lemma}{Lemma}

        \newcommand{\tr}{\mathrm{tr}}
        \newcommand{\Cov}{\mathrm{cov}}
        \newcommand{\Var}{\mathrm{var}}
        \newcommand{\Normal}[2]{\mathcal{N}\left(#1,#2 \right)}
        \newcommand{\Normall}[2]{\mathcal{N}(#1,#2)}
        
        \newcommand{\undertext}[2]{\ensuremath{\underset{#2}{\underbrace{#1}}}}
        \newcommand{\changed}[1]{{#1}}

\modulolinenumbers[5]

\begin{document}

\begin{frontmatter}

\title{Probabilistic partial least squares model:\\ Identifiability, estimation and application}


\author[LUMC]{Said el Bouhaddani\corref{cor}}
\cortext[cor]{Corresponding author}
\ead{S.el\_ Bouhaddani@lumc.nl}

\author[Utrecht,LUMC]{Hae-Won Uh}
\ead{H.W.Uh@umcutrecht.nl}

\author[Edin]{Caroline Hayward}
\ead{Caroline.Hayward@igmm.ed.ac.uk}

\author[TU]{Geurt Jongbloed}
\ead{G.Jongbloed@tudelft.nl}

\author[Leeds,LUMC]{Jeanine Houwing-Duistermaat}
\ead{J.Duistermaat@leeds.ac.uk}

\address[LUMC]{Department of Medical statistics and bioinformatics, Leiden University Medical Center, The Netherlands}
\address[Utrecht]{Department of Biostatistics and Research Support UMC Utrecht, div. Julius Centrum, University Medical Center Utrecht, The Netherlands}
\address[TU]{Department of Applied Mathematics, Delft University of Technology, The Netherlands}
\address[Leeds]{Department of Statistics, University of Leeds, United Kingdom}
\address[Edin]{MRC Human Genetics Unit, Institute of Genetics and Molecular Medicine, University of Edinburgh, Scotland}

\begin{abstract}
With a rapid increase in volume and complexity of data sets, there is a need for methods that can extract useful information, for example the relationship between two data sets measured for the same persons. 
The Partial Least Squares (PLS) method can be used for this dimension reduction task.
\changed{Within life sciences, results across studies are compared and combined. Therefore, parameters need to be identifiable, which is not the case for PLS. In addition, PLS is an algorithm, while epidemiological study designs are often outcome-dependent and methods to analyze such data require a probabilistic formulation. Moreover, a probabilistic model provides a statistical framework for inference.
To address these issues, we develop Probabilistic PLS (PPLS).}
We derive maximum likelihood estimators that satisfy the identifiability conditions by using an EM algorithm with a constrained optimization in the M step. We show that the PPLS parameters are identifiable up to sign. 
A simulation study is conducted to study the performance of PPLS compared to existing methods. The PPLS estimates performed well in various scenarios, even in high dimensions. Most notably, the estimates seem to be robust against departures from normality. 
To illustrate our method, we applied it to IgG glycan data from two cohorts. Our PPLS model provided insight as well as interpretable results across the two cohorts.

\end{abstract}

\begin{keyword}
Dimension reduction \sep EM algorithm \sep Identifiability \sep Inference \sep Probabilistic partial least squares
\end{keyword}
 
\verb|© 2018. This manuscript version is made available under the CC-BY-NC-ND 4.0 license|\\
\verb|http://creativecommons.org/licenses/by-nc-nd/4.0/|

\end{frontmatter}


\section{Introduction}
With the exponentially growing volume of data sets, multivariate methods for reducing dimensionality are an important research area in statistics. For combining two data sets, Partial Least Squares (PLS) regression \cite{Wold1973} is a popular dimension reduction method \cite{Abdi2010}. PLS decomposes variation in each data set in a joint part and a residual part. The joint part is a linear projection of one data set on the other that best explains the covariance between the two data sets. These projections are obtained by iterative algorithms, such as NIPALS \cite{Wold1973}. 
\changed{
Partial Least Squares is popular in chemometrics \cite{Boulesteix2007}. In this field, the focus is on development of algorithms with good prediction performance, while the underlying model is less important. For applications in life sciences, interpretation of parameter estimates is necessary to gain understanding of the underlying molecular mechanisms.
}

\changed{
For interpretation, a model needs to be identifiable. A model is said to be unidentifiable if the model corresponds to more than one set of parameter values. For PLS, rotation of the parameters does not change the model \cite{Wang2005}. Hence, PLS does not provide an identifiable model. By constraining the parameter space, identifiability can be obtained. This involves solving a challenging optimization problem, since PLS requires estimating a structured covariance matrix \cite{Ros2016}.
}

\changed{
For many problems in life sciences the study design needs to be accounted for, and algorithmic approaches such as PLS cannot be applied. Hence, a probabilistic formulation is necessary. Since likelihood method provides asymptotic standard errors of parameter estimates, computer-intensive resampling procedures can be avoided. 
}

Also for other dimension reduction techniques, probabilistic methods have been developed. In 1999, Tipping and Bishop \cite{Tipping1999} developed the Probabilistic Principal Component Analysis (PPCA), in order to deal with missing data and dependent samples. 
In 2005, Bach and Jordan \cite{Bach2005} developed Probabilistic Canonical correlation analysis (PCCA). However, for both PPCA and PCCA the model parameters are not identifiable, since rotation of the parameters does not change the model \cite{Tipping1999, Bach2005}. \changed{In addition, in 2015, simultaneous envelopes models have been developed \cite{Cook2015} for `low-dimensional' settings. Further, Probabilistic PLS Regression and Probabilistic PLS have been proposed \cite{Li2015,Zheng2016}. For all these approaches, the model parameters are not identifiable.
}

\changed{
In this paper we propose the Probabilistic Partial Least Squares (PPLS) model and show that the model parameters are identifiable up to a sign. We propose to maximize the PPLS likelihood with an EM algorithm that decouples the likelihood into several factors involving distinct sets of parameters. In the M step, a constrained optimization problem is solved by using a matrix of Lagrange multipliers. 
}

The rest of the paper is organized as follows: In Section \ref{sec:Model} we develop the PPLS model and establish identifiability of the model parameters. We develop an efficient algorithm for estimating the PPLS parameters. In Section \ref{sec:simu} we study the performance of the PPLS estimators via simulations. In Section \ref{sec:data} we illustrate the PPLS model with two data matrices from two cohorts. We finish with a discussion. 

\section{Model and estimation}
\label{sec:Model}
\subsection{The PPLS model}
Let $x$ and $y$ be two random row-vectors of dimension $p$ and $q$, respectively. The Probabilistic Partial Least Squares (PPLS) model describes the two random vectors in terms of a joint part and a noise part. The joint part consists of correlated latent vectors, denoted by $t$ and $u$, while the noise part consists of isotropic normal random vectors referred to as $e$, $f$ and $h$. The dimension of $t$ and $u$ is denoted by $r$. The PPLS model describing the relationship between $x$, $y$ and the joint and noise parts is
\begin{equation}
x =  tW^\top + e, \quad
y =  uC^\top + f, \quad
u = tB + h.
\label{eq:PPLS}
\end{equation}

Specifically, $e=(e_1,\ldots,e_p)$, $f=(f_1,\ldots,f_q)$ and $h=(h_1,\ldots,h_r)$ are independent with zero mean and referred to as {noise variables}. The distributions of $e$, $f$ and $h$ are multivariate normal with positive definite covariance matrix proportional to the identity matrix,
\begin{equation*}
e \sim  \Normall{0}{\sigma_e^2 I_p}, \quad
f \sim  \Normall{0}{\sigma_f^2 I_q}, \quad
h \sim  \Normall{0}{\sigma_h^2 I_r}.
\end{equation*}

The latent vector $t=(t_1,\ldots,t_r)$ is an $r$-dimensional multivariate normal vector with with zero mean and diagonal positive definite covariance matrix $\Sigma_t = \mathrm{diag}(\sigma^2_{t_1},\ldots,\sigma^2_{t_r})$, so 
\begin{equation}
t \sim \Normal{0}{\Sigma_t}.
\label{eq:Tindep}
\end{equation}
The matrix $B=\mathrm{diag}\left(b_1,\ldots,b_r\right)$ is a diagonal matrix of size $r$, containing regression coefficients of $u$ on $t$.
Finally $W$ ($p\times r$) and $C$ ($q\times r$) are parameter matrices, referred to as {loadings}. 
The PPLS model for the random $p$-dimensional row-vector $x$ and random $q$-dimensional row-vector $y$ is given in Eq.~\eqref{eq:PPLS}. \changed{Let $\theta$ be the parameters of the PPLS model, i.e.,}
\begin{equation}
\label{eq:theta}
\theta = (W, C, B, \Sigma_t, \sigma_e, \sigma_f, \sigma_h).
\end{equation}
The PPLS model and its parameters are formulated conditional on the value of the dimension of the latent space $r$.

The PPLS model \eqref{eq:PPLS} assumes a multivariate normal distribution for the observable random vectors $x$ and $y$. The covariance between $x$ and $y$ is modeled by the regression of the latent vector $u$ on $t$. The distribution of $(x,y)$ is $\Normal{0}{\Sigma}$ with density given, for ${x}\in\mathbb{R}^p$ and ${y}\in\mathbb{R}^q$, by
\begin{equation*}
f({x},{y}) = (2\pi)^{-(p+q)/2} |\Sigma|^{-\frac{1}{2}} \e^{({x},{y})\Sigma^{-1}({x},{y})^\top},
\end{equation*}
and covariance matrix
\begin{equation}
\Sigma = \begin{pmatrix}
\Sigma_x & \Sigma_{x,y} \\
\Sigma_{y,x} & \Sigma_y
\end{pmatrix} = \begin{pmatrix}
W \Sigma_t W^\top + \sigma_e^2 I_p & W \Sigma_t B C^\top \\
C B \Sigma_t W^\top & C (B^2 \Sigma_t + \sigma_h^2 I_r) C^\top + \sigma_f^2 I_q
\end{pmatrix}.
\label{eq:covmatrix}
\end{equation}
This follows from the normality property and from computing the variances and covariances of the random vectors; see Appendix \ref{ap:theoCov} for the details.

\subsection{Identifiability of probabilistic PLS}
To establish identifiability of the PPLS model, some assumptions about its parameters have to be made. First, we assume that $ 0 < r < \min(p,q)$. Second, we assume that the diagonal elements of $B$ are positive, $b_k>0$ for $k\in\{1,\ldots,r\}$. This will not restrict the model, since $t_k b_k$ is equal to $-t_k b_k$ in distribution. To identify the order of the loading vectors, the elements of $(\sigma_{t_k}^2 b_k)_{k=1}^r$ are assumed to be strictly decreasing with $k$. Finally, we assume that the loading matrices $W$ and $C$ are orthogonal, i.e.,  $W^\top W = C^\top C = I_r$. Together with the diagonality of $\Sigma_t$ in \eqref{eq:Tindep}, it implies identifiability of all parameters up to sign. This is shown in the following theorem.
\begin{theorem}\label{th:identif}
Let $r$ be fixed such that $0 < r < \min(p,q)$.
Let $(x,y)_1$ and $(x,y)_2$ be generated by the PPLS model \eqref{eq:PPLS} having covariance matrix $\Sigma_1$ and $\Sigma_2$ with underlying parameters $\theta_1$ and $\theta_2$ as defined in \eqref{eq:theta}, respectively. 
Then $
\Sigma_1 = \Sigma_2 $
implies that $W_1 = W_2 \Delta$, $C_1 = C_2 \Delta$ for some diagonal matrix $\Delta$ with on the diagonal elements $\delta_i \in \{-1,1\}$, for $i\in\{1,\ldots,r\}$, and all other parameters in $\theta_1$ and $\theta_2$ are equal.
\end{theorem}

The formal proof is given in Appendix \ref{ap:identif}. 
Identifiability up to sign can be represented by a diagonal orthogonal matrix, namely a diagonal matrix with diagonal elements in $\{-1,1\}$. 
For example, taking the model for $x$ in \eqref{eq:PPLS}, we may substitute $W$ by $WR_S$ and $t$ by $t R_S$, where $R_S$ is a diagonal orthogonal matrix, and get
\begin{equation*}
x = t R_S R_S^\top W^\top + e = \sum_{j=1}^r t_j (R_S)_{jj}^2 w_j^\top + e.
\end{equation*}
Since $(R_S)_{jj}^2 = 1$ and the distribution of $t_j$ and $-t_j$ is the same, the right-hand side reduces to the original model for $x$ in \eqref{eq:PPLS}.
Note that the PPLS model is not invariant under general rotation matrices. Take a general rotation matrix $R$, then we still get 
\begin{equation*}
x = t W^\top + e = t R R^\top W^\top + e,
\end{equation*}
since $RR^\top = I_r$. Inspecting the covariance of $TR$ we see that $\Cov(TR) = R^\top \Sigma_t R$, which is not diagonal if $R$ is not diagonal, and violates the PPLS model assumption on $\Sigma_t$ in Eq.~\eqref{eq:Tindep}.

\subsection{Estimating the parameters}
\label{subsec:estim}
\changed{
Unlike the iterative PLS methods, we propose a simultaneous approach for estimating the parameters, while taking the constraints in the PPLS model into account. 
Given the number of PPLS components, $r$,} the log likelihood of an independent and identically distributed (iid) sample $(X,Y)=\{(X_1,Y_1)^\top,\ldots,(X_N,Y_N)^\top\}^\top$ of size $N$ from $(x,y)$ is
\begin{equation}
L(\theta) = - \frac{N(p+q)}{2} - \frac{N}{2}\ln|\Sigma| - \frac{N}{2}\tr (S\Sigma^{-1} )
\label{eq:loglike_obs}
\end{equation}
with $S = N^{-1}\sum_{i=1}^N (X_i,Y_i)^\top (X_i,Y_i)$ and $\Sigma$ as in Eq.~\eqref{eq:covmatrix}. To ensure empirical identifiability, we assume that $r < N$. Note that the data dimensionality $p$ and $q$ may be larger than $N$.
For estimation of $\theta$, maximum likelihood is used. 

The log likelihood \eqref{eq:loglike_obs} depends in a non-linear way on the theoretical covariance matrix $\Sigma$, which contains the loadings and variances. Optimizing this function directly is a non-trivial task, especially in high dimensions (i.e. when $p$ and $q$ are large). However, the PPLS model allows for a more simple (but iterative) optimization approach. 
Indeed, the maximum likelihood estimates for $\theta$ are a least squares type solution if the latent variables $t$ and $u$ are observed, as the model for $x$ and $y$ in \eqref{eq:PPLS} involves known $t$ and $u$.
In contrast, knowing $\theta$ allows for reconstruction of $t$ and $u$ by computing their conditional means given $x$ and $y$. Alternating these two scenarios is actually an Expectation-Maximization (EM) \cite{Dempster1977} algorithm, with observed data $(x,y)$ and missing data $(t,u)$. 

\paragraph{The EM algorithm}
The joint distribution of the complete data $(x,y,t,u)$ can (with abuse of notation) be decomposed as
\begin{equation}
f(x,y,t,u) = f(x|t) f(y|u) f(u|t) f(t).
\label{eq:decomposition}
\end{equation}
This follows from
\begin{equation*}
f(x,y,t,u) = f(x,y|t,u)f(t,u) = f(x|t,u)f(y|t,u)f(t,u).
\end{equation*}
The second equation is implied by the fact that $x$ and $y$ are independent given $t$ and $u$. The first two factors in the right-hand side can be rewritten as $f(x|t,u) = f(x|t)$ and $f(y|t,u) = f(y|u)$, since $x$ and $u$ are independent given $t$, and $y$ and $t$ are independent given $u$. The last factor can be rewritten as $f(u|t)f(t)$, yielding Eq.~\eqref{eq:decomposition}.
The logarithm of the first three factors in the product in \eqref{eq:decomposition} can be written as
\begin{equation*}
\begin{split}
\ln f(X|T) & = -\frac{Np}{2\pi\sigma_e^2} - \frac{1}{2\sigma_e^2} \sum_{i=1}^N ||X_i - T_iW^\top||^2, \\
\ln f(Y|U) & = -\frac{Nq}{2\pi\sigma_f^2} - \frac{1}{2\sigma_f^2} \sum_{i=1}^N ||Y_i - U_iC^\top||^2, \\
\ln f(U|T) & = -\frac{Nr}{2\pi\sigma_h^2} - \frac{1}{2\sigma_h^2} \sum_{i=1}^N ||U_i - T_iB||^2.
\end{split}
\end{equation*}
Denote by $L_{\rm Comp} = \ln f(X,Y,T,U)$ the complete data log-likelihood, and define
\begin{equation*}
Q(\theta) = \mathrm{E} \{L_{\rm Comp}(\theta)|X,Y,\theta' \},
\end{equation*}
where the expectation is taken conditional on the observed $X$ and $Y$, and $\theta'$ is a fixed current estimate of the parameters. 
By optimizing $Q$ over all allowed $\theta$, we get a non-negative increase in the \emph{observed} log-likelihood $L$. Moreover, by iterating this process of taking the expectation and maximizing over $\theta$, the estimates in general converge to a stationary point or, in particular, a (possibly local) maximum of $L$ \cite{Dempster1977,Wu1983}.
The expectation step calculates the conditional expectation of the missing data given the observed data given by $Q(\theta)$, which may in general involve intractable integration. However, for the exponential family, in particular the multivariate normal family, the complete likelihood depends on the complete data only via the sufficient statistics (called $t(\mathbf{x})$ in \cite{Dempster1977}), which are given in terms of the first and second moments of the complete data for the multivariate normal distribution. Computing $Q(\theta)$ implies computing the expected first and second moment of the latent variables: $\mathrm{E}\left( T | X,Y,\theta \right)$, $\mathrm{E}\left( T^\top T | X,Y,\theta \right)$, $\mathrm{E}\left( U | X,Y,\theta \right)$,  $\mathrm{E}\left( U^\top U | X,X,\theta \right)$ and $\mathrm{E}\left( U^\top T | X,Y,\theta \right)$; see Appendix \ref{ap:EM} for details. Moreover, the decomposition in \eqref{eq:decomposition} allows for optimization of $\mathrm{E}\{\ln f(X|T)\}$, $\mathrm{E}\{\ln f(Y|U)\}$ and $\mathrm{E}\{\ln f(U|T)\}$ separately, while only considering parameters involved in each factor. Maximizing $Q$ over $\theta$ yields parameter estimates for the next iteration in the EM algorithm.
This leads us to the following theorem.
\begin{theorem}
Let $X$ and $Y$ be an observed data sample of size $N$, generated according to the PPLS model \eqref{eq:PPLS}. Let $r$ be fixed such that $0 < r < \min(N,p,q)$.
The parameters in $\theta$ can be estimated with an EM algorithm, yielding the following iterative scheme in $k$ with given starting values for $k=0$:
\begin{equation*}
\begin{gathered}
W^{k+1} = X^\top \, \mathrm{E} ( T | X,Y,\theta^k  ) (L_W^\top )^{-1}; \quad
C^{k+1} = Y^\top \, \mathrm{E} ( U | X,Y,\theta^k  ) (L_C^\top )^{-1}; \\
B^{k+1} = \mathrm{E} ( U^\top T | X,Y,\theta^k  )  \{\mathrm{E} ( T^\top T | X,Y,\theta^k  ) \}^{-1} \circ I_r; \;
\Sigma_{t}^{k+1} = \frac{1}{N}\mathrm{E} ( T^\top T | X,Y,\theta^k ) \circ I_r; \;
(\sigma^2_{h})^{k+1} = \frac{1}{Nr}\tr\mathrm{E} ( H^\top H | X,Y,\theta^k  ); \\
(\sigma^2_{e})^{k+1} = \frac{1}{Np}\tr\mathrm{E} ( E^\top E | X,Y,\theta^k  ); \quad
(\sigma^2_{f})^{k+1} = \frac{1}{Nq}\tr\mathrm{E} ( F^\top F | X,Y,\theta^k  );
\end{gathered}
\end{equation*}
where $L_W$ and $L_C$ are such that
\begin{equation*}
L_W L_W^\top = \mathrm{E} ( T^\top | X,Y,\theta^k  ) \, X \, X^\top \, \mathrm{E} ( T | X,Y,\theta^k  ), \quad
L_C L_C^\top = \mathrm{E} ( U^\top | X,Y,\theta^k ) \, Y \, Y^\top \, \mathrm{E} ( U | X,Y,\theta^k  ).
\end{equation*}
\label{th:EM}
\end{theorem}

The proof for Theorem \ref{th:EM} and the expressions for the conditional expectations are given in Appendix \ref{ap:EM}. 
Note the dependency of $W^{k+1}$ and $C^{k+1}$ on the matrices $L_W$ and $L_C$. These matrices ensure orthogonality of $W^{k+1}$ and $C^{k+1}$ in each iteration:
\begin{equation*}
 (W^{k+1} )^\top W^{k+1} = L_W^{-1} \, \tilde{T}^\top \, X \, X^\top \, \tilde{T}  (L_W^\top )^{-1} = L_W^{-1}\,L_W\,L_W^\top\, (L_W^\top )^{-1} = I_r,
\end{equation*}
where $\tilde{T} = \mathrm{E}( T | X,Y,\theta^k )$. The exact forms of $L_W$ and $L_C$ are not unique. Two choices are the eigenvectors of $\mathrm{E}( T^\top | X,Y,\theta^k ) \, X \, X^\top \, \mathrm{E} ( T | X,Y,\theta^k  )$ and the lower triangular matrix of $X^\top \, \mathrm{E}( T | X,Y,\theta^k )$ in the Cholesky decomposition. Note that these two orthogonalization matrices are straightforward to calculate with standard linear algebra tools. Since the PPLS model is identifiable, all choices for $L_W$ and $L_C$ will lead to the same optimum as the iteration number $k$ tends to infinity. 

\paragraph{Standard errors for PPLS}
\changed{
Asymptotic standard errors for maximum likelihood estimators are found by inverting the observed Fisher information matrix. Following the reasoning of \cite{Louis1982}, the observed information may be given by
\begin{equation*}
\mathrm{E} \{ B(\hat{\theta}) | X, Y \} - \mathrm{E} \{ S(\hat{\theta})S(\hat{\theta})^\top | X, Y \}.
\end{equation*}
Here $S(\hat{\theta}) = \nabla \lambda(\hat{\theta})$ and $B(\theta) = -\nabla^2 \lambda(\hat{\theta})$ are the gradient and negative of the second derivative of the log likelihood $\lambda(\theta)$ respectively evaluated in the MLE $\hat{\theta}$. 
The explicit form of the asymptotic covariance matrix of $w_k$ is given in Appendix \ref{ap:SE}. The square root of the diagonal elements are the asymptotic standard errors for the corresponding loading estimates.}

\paragraph{Finding the number of components $r$}
\changed{
Available approaches to determine the number of PPLS components $r$ are minimizing a cross-validated loss function \cite{Geisser1993}, visually inspecting eigenvalues of a covariance matrix \cite{Mardia1979}, and selecting the number of components needed to achieve a certain proportion of variance explained by the components. In this paper we apply the last approach.
}

The PLS and PPLS algorithms are available as \textsf{R} packages at \url{github.com/selbouhaddani} under repository OmicsPLS and PPLS, respectively.

\section{Simulation study}
\label{sec:simu}
\changed{
To evaluate the performance of the PPLS estimates, a simulation study was conducted. The aim was (1) to investigate the performance of PPLS for various scenarios, (2) to evaluate robustness of the PPLS estimates against departures from the normality assumption, (3) to compare the performance of the loading estimates with other probabilistic approaches, and (4) to compare the asymptotic PPLS standard errors with the bootstrap standard errors.
}

The simulated data were generated according to the PPLS model \eqref{eq:PPLS}. The number of components was chosen to be $3$, both in the data generation and in the estimation. We considered combinations of small and large sample size ($N\in\{50,500\}$), low and high dimensionality ($p\in\{20,1000\}$), and small and large proportion of noise (denoted by $\alpha_n \in\{0.1,0.5\}$). The robustness of PPLS was evaluated by considering four different continuous and discrete distributions for the latent variables $t$, $u$, $e$, $f$ and $h$; we chose the normal distribution, the $t$ distribution with two degrees of freedom, the Poisson distribution with rate 1, and the Binomial distribution with two trials and success probability 0.25. These distributions cover a wide range of characteristics typically observed in omics data: heavy tailed, skewed and/or discrete. The latent variables were scaled to have mean zero and variances as specified below.
All scenarios are summarized in Table \ref{tb:simulation_scenarios}.

\changed{
The true loading values per component were generated from the normal density function with parameters $\mu$ and $\sigma$, denoted by $N(x; \mu, \sigma^2)$, as follows
\begin{equation*}
w_{j,k} = N\{j; (1/2 + 1/10j)p, 1/10p\}, \quad
c_{j,k} = N\{j; (3/5 + 1/10j)q, 1/10q\}.
\label{eq:SimuLoad}
\end{equation*}
The second columns in $W$ and $C$ were orthonormalized with respect to the first columns, and the third columns were orthonormalized with respect to the first two columns; we used a Gram--Schmidt procedure for both operations. 
The elements of the diagonal matrix $B$ were set to $b_{k} = \e^{\ln(1.5)-3(k-1)/10} = (1.5, 1.11, 0.82)$, for $\Sigma_t$ we chose $\sigma_{t_k} = \e^{-(k-1)/10} = (1, 0.90, 0.82)$. 
}

For comparing the parameter estimates with the true values $\theta$, we computed the bias and the variance of the estimates.
To deal with the identifiability up to sign, we multiplied each estimated loading vector by $-1$ if the inner product of the estimated loading vector and the true loading vector was negative. Moreover, we swapped columns in $W$ and $C$ to maintain the same ordering as the ordering in the true loadings. This was done to avoid inflation of the bias or variance due to a wrong sign or ordering of the individual components. 

\changed{
PPLS estimates were compared to PLS estimates (with orthogonal loadings, see \cite{Rosipal2005} for an overview) for all scenarios above. For comparing PPLS with PPCA and PCCA, we constructed a `null model', i.e.,  $B=0$, as well as $B\neq 0$. We used the same scenarios as above, but we only considered the normal distribution.
}

\begin{table}
\caption{{Overview of the simulation scenarios.} The noise level is defined as the proportion of variation in the noise matrices $E$, $F$ and $H$ relative to the total variation in $X$, $Y$ and $U$ respectively.}
\begin{tabularx}{\textwidth}{l||X}
{Sample size}	& $N = (50,\ 500)$ 										\\
{Dimensionality}	& $p = q = (20,\ 1000)$										\\
{Noise level}	& $\alpha_n = (0.1,\ 0.5)$									\\
{Distribution of $t$, $u$, $e$, $f$ and $h$}	& $\{\mathcal{N}(0,1),\ t_2,\ \mathcal{P}(1),\ \mathcal{BIN}(2, 0.3)\}$	\\
\end{tabularx}
\label{tb:simulation_scenarios}
\end{table}

\changed{
Regarding standard errors for PPLS loadings, we compared asymptotic standard errors (as in Section \ref{sec:Model}) and bootstrap standard errors \cite{Wehrens1997}. One set of two data matrices $X$ and $Y$ was simulated from a PPLS model with $p=q=20$ normally distributed variables. Based on these data, asymptotic and bootstrap standard errors were calculated. The number of bootstrap replicates was 1000. Furthermore, simulation-based standard errors for the loadings (based on standard deviations over 1000 data matrices drawn from the PPLS model used to generate the original data) were included as reference.
Low and high noise levels ($\alpha_n=0.1$ resp. $\alpha_n=0.5$), and small, large and `extra large' sample sizes ($N=50$, $N=500$ and $N=5000$, respectively) were considered. In the `extra large' sample size scenario, no simulation-based reference was calculated. The PPLS estimation algorithm was considered to be converged when either the log-likelihood increment was below $10^{-6}$, or $10^4$ EM steps were made. For each scenario, $1000$ replicates are used. 
}

\subsection{Results}
\paragraph{Results for the loadings}
The biases and variances \changed{of the estimated first component $W_1$} for the low dimensional case for normally distributed latent variables are graphically depicted in \changed{Figure \ref{fig:normlow}}. A black dot represents the average estimated PPLS loading value across 1000 simulations, whereas the width of the black dashed vertical line equals two times the standard deviation across 1000 simulations. The red star and red dashed vertical line represent the average loading value and twice the standard deviation for the PLS estimates. The true loading values are represented by a step function with steps at each index $j\in\{1,\ldots,p\}$. \changed{Results for other components and scenarios are included in the Online Supplement.}

Comparing the estimates for the \textit{first} loading component $W_1$, a better performance of PPLS compared to PLS was observed in terms of bias. In all scenarios the bias of the PPLS estimators were about the same as or less than the bias of the PLS estimators. Both estimators showed larger bias towards zero for higher absolute loading values. The biases decreased with a larger sample size and lower noise level. The biases of both estimators were very similar across different distributions. 
In the scenario where there is 50\% noise and few (50) samples the variance of the PPLS estimators tended to be slightly larger than the variances of the PLS estimators when the true loading values were larger. This was observed across all distributions. The variances of the PPLS estimates were about the same or lower than the PPLS estimates in all other scenarios, where either the noise level was less or more samples were available. For both PPLS and PLS estimators the variances tended to increase with higher loading values. The variances decreased with larger sample size and lower noise level. The variances of bots estimators were very similar across different distributions.
For the loading component $C_1$ and their PLS and PPLS estimators the same conclusions were obtained.

For the \textit{second} loading component $W_2$ \changed{(shown in the Online Supplement)}, the biases of the PPLS loading estimates were as good as, and often better than the PLS loading estimates, especially at lower values. In the scenarios of 50\% noise and a small sample size ($N=50$) the bias was slightly larger for PPLS estimators compared to PLS estimates when the loading values were larger. 
Both estimators showed larger bias towards zero for higher loading values. The biases decreased with a larger sample size and lower noise level. For all distributions, the biases of both estimators were very similar. 
The variances of the PPLS estimators were as good as or lower than the PLS estimators, except in the scenario in which both the noise level was high (50\%) and the sample size was small (50). In this scenario the variances of the PPLS estimators were still lower if the true loading values were close to zero, but higher for higher loading values. 
For both PPLS and PLS estimators the variances tended to increase with higher loading values. The variances decreased with larger sample size and lower noise level. The variances of both estimators were very similar across different distributions. 
For the loading component $C_2$ and their PLS and PPLS estimators the same conclusions were obtained.

For the \textit{third} loading components $W_3$ and $C_3$ \changed{(shown in the Online Supplement)}, the same observations were made as for the first loading components $W_1$ and $C_1$, both for the biases as for the variances.

For the high and extra high-dimensional case, the same results were obtained for the loadings $W$ and $C$. \changed{See the Online Supplement for more details.}

With regard to the comparison of PPLS with PPCA and PCCA, PPLS performed better than PCCA and similar to PPCA in most scenarios. \changed{Details are given in the Online Supplement.}

\paragraph{Results for the variance parameters}
The performance of the estimators of the variance parameters $B$, $\sigma_t$, $\sigma_e$, $\sigma_f$ and $\sigma_h$ were also evaluated, the results are shown in Figure \ref{fig:normlowV}. We did not compare with PLS as these model parameters are not present in the PLS framework. For sake of comparison, we calculated the relative biases and variances of the estimates with respect to the true corresponding parameter value.
The biases of the PPLS estimators for all variance parameters were very small for large sample size ($N=500$), regardless of the noise. For small sample size ($N=50$), the first two diagonal elements of $B$ and $\Sigma_t$ were overestimated, while the last component was underestimated. The noise parameters $\sigma_e$ and $\sigma_f$ were underestimated in these scenarios, while the estimator for $\sigma_h$ was unbiased, except in combination with a low noise level (10\%). The relative biases decreased slightly with lower noise level, except for the earlier mentioned $\sigma_h$, and decreased more with larger sample size.
The relative variances of the estimators of $B$, $\Sigma_t$ and $\sigma_h$ were larger than the variances of the estimators of $\sigma_e$ and $\sigma_f$. 
For $B$, there was a slight increase in variance across the three components. The variances decreased slightly with lower noise and more with larger sample size. The variances slightly decreased in the scenario of high dimensionality and high noise level.
The same observations were made across the different distributions.

\paragraph{Ordering of the loadings}
We compared the ordering of the true loadings $W$ and $C$ with the ordering of the estimated loadings. This provides a proportion across the 1000 simulation replicates in which the ordering matched. 
In Table \ref{tab:mismatchprop}, the proportion of correct orderings of $W$ for the scenario with normally distributed latent variables is shown for different scenarios. It can be seen that the proportion of correct orderings tends to be lower with smaller sample size and with higher noise level. Moreover, if the sample size is small, the proportion of correct orderings is much lower with higher noise. 
A higher dimensionality has a slightly negative impact on the correct ordering proportion when the sample size is larger, but a positive impact in the small sample size scenario. Especially, when also the noise level is high, this can be considerable. The same observations were made for the other distributions. Exactly the same proportions were observed for the loadings $C$.

\begin{table}
\caption{{Proportion of correct order of loadings $W$ and $C$ across 1000 simulation replicates.} These were obtained for different values of the dimensionality (high = 1000 variables, low = 20 variables), sample size (large = 500 subjects, small = 50 subjects) and noise level (high = 50\% noise, low = 10\% noise).}
\centering
\begin{tabularx}{.7\textwidth}{c|c|c|c}
{Dimensionality} & {Sample Size} & {Noise Level} & {Correct Ordering Proportion} \\
\toprule
\multirow{4}{*}{low}  & \multirow{2}{*}{large} & low  & 1.000 \\
  &  & high & 0.989 \\
  & \multirow{2}{*}{small}   & low  & 0.932 \\
  &   & high & 0.435 \\
\midrule
\multirow{4}{*}{high} & \multirow{2}{*}{large} & low  & 0.990 \\
 &  & high & 0.985 \\
 & \multirow{2}{*}{small}   & low  & 0.940 \\
 &   & high & 0.665 \\
\bottomrule
\end{tabularx}
\label{tab:mismatchprop}
\end{table}

\paragraph{Comparison of PPLS standard errors}
\changed{
The results for low noise level are shown in Figure \ref{fig:W_SE_nlow}. In all scenarios, the asymptotic standard errors were smaller than the bootstrap standard errors for nearly all loading elements. In particular, for high loading values the difference between asymptotic and bootstrap standard errors tended to be large. This difference decreased with larger sample size: In the `extra large' sample size, the bootstrap and asymptotic standard errors had similar magnitude. Similar observations were made for other distributions. For details, see the Online Supplement.
}

\begin{figure}
\begin{subfigure}[h]{.5\linewidth}
\centering
\includegraphics[width=\textwidth]{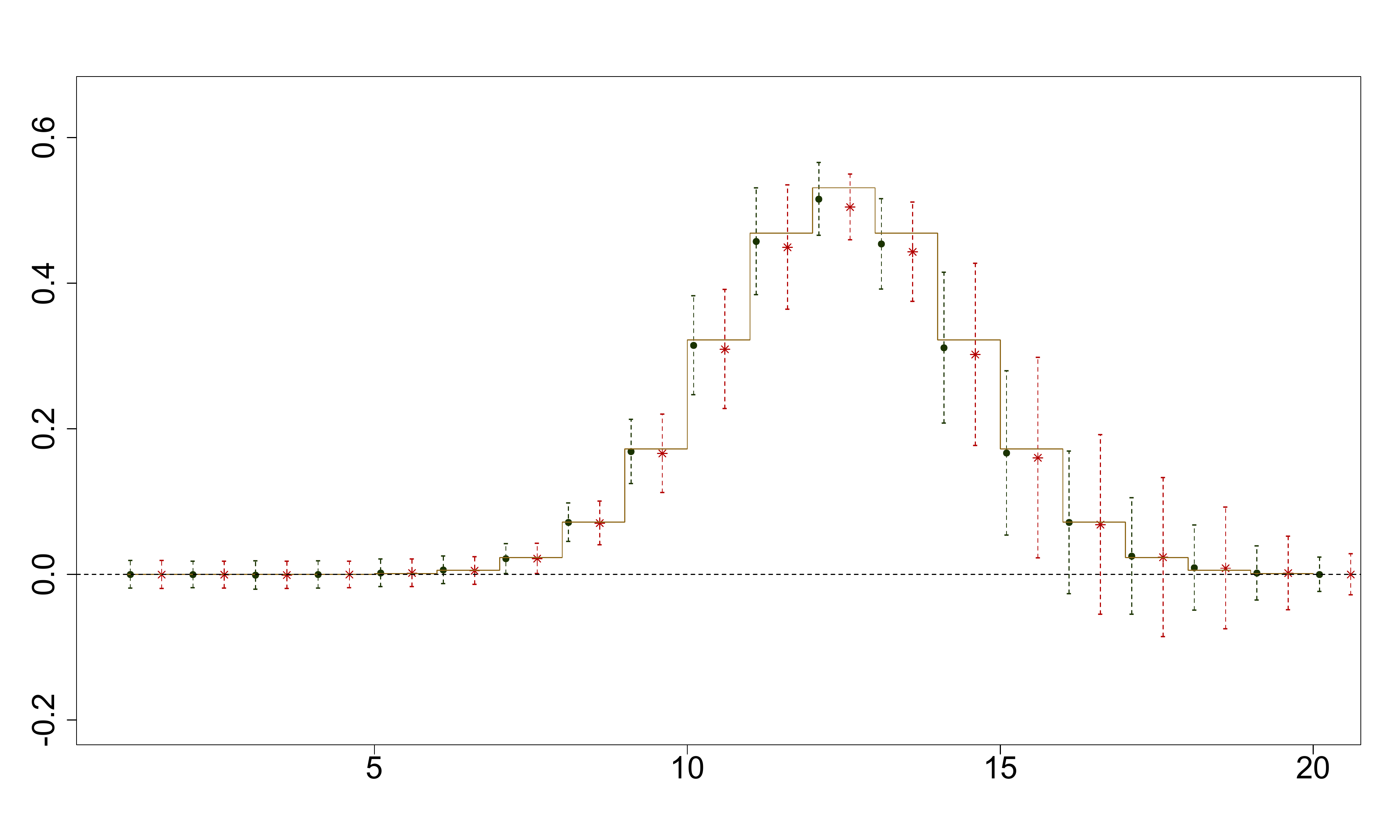}
\caption{Low noise (10\%); small sample size ($N=50$)}
\label{fig:normlow_lowlow}
\end{subfigure}%
\begin{subfigure}[h]{.5\linewidth}
\centering
\includegraphics[width=\textwidth]{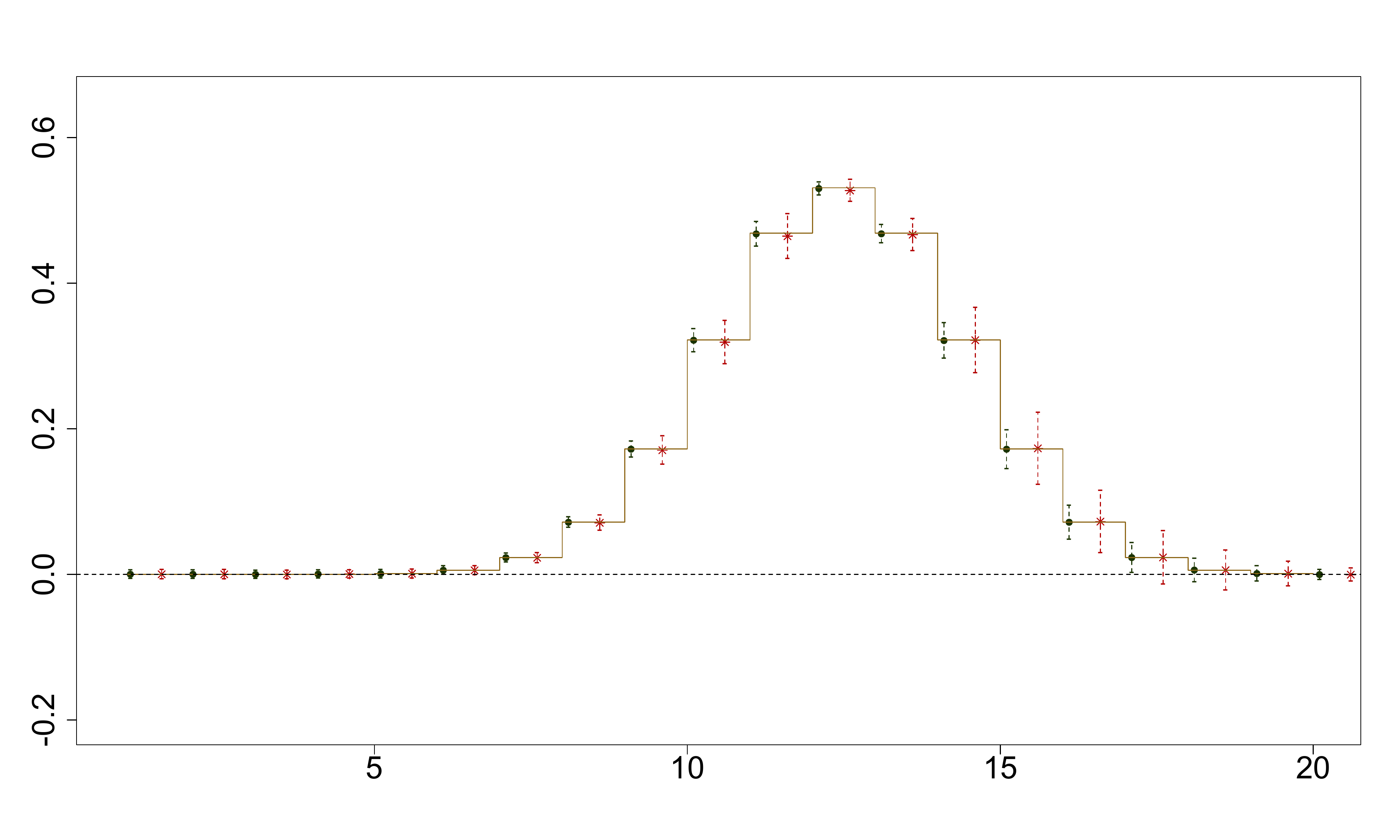}
\caption{Low noise (10\%); large sample size ($N=500$)}
\label{fig:normlow_lowhigh}
\end{subfigure}
\\
\begin{subfigure}[h]{.5\linewidth}
\centering
\includegraphics[width=\textwidth]{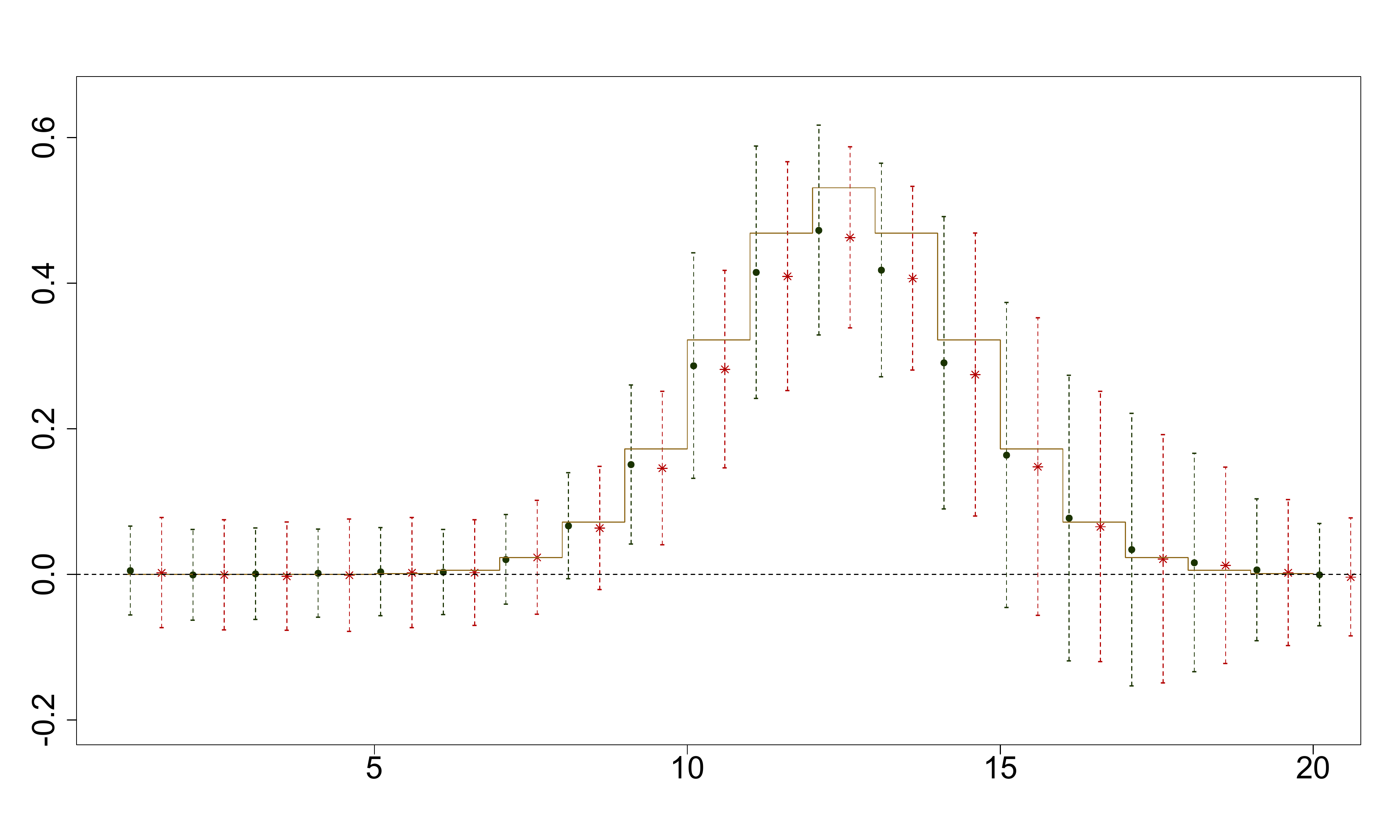}
\caption{High noise (50\%); small sample size ($N=50$)}
\label{fig:normlow_highlow}
\end{subfigure}
\begin{subfigure}[h]{.5\linewidth}
\centering
\includegraphics[width=\textwidth]{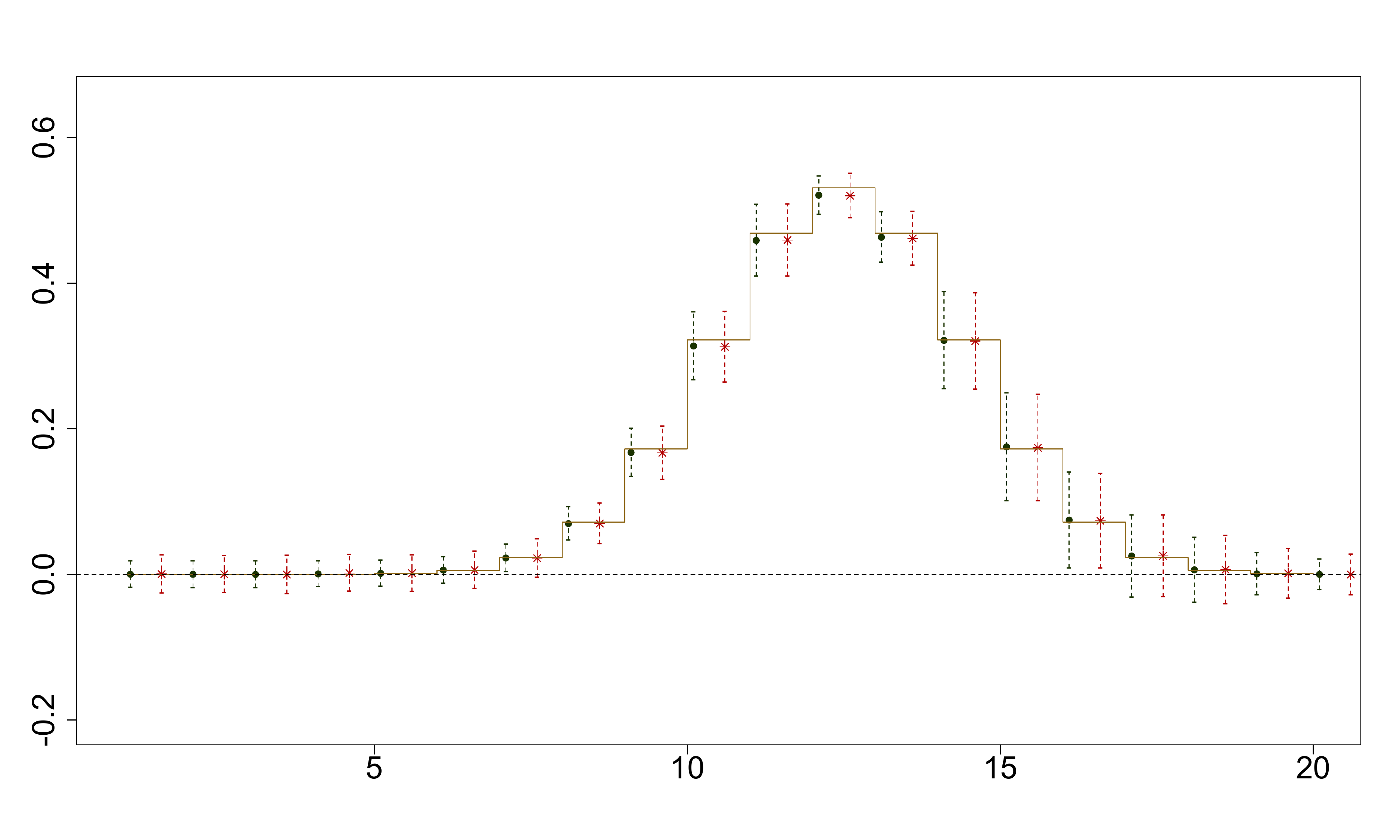}
\caption{High noise (50\%); large sample size ($N=500$)}
\label{fig:normlow_highhigh}
\end{subfigure}
\caption{\changed{{True and estimated loadings $W_1$ over 1000 simulation replications.}} The black dots and dashed vertical lines (on the left of each pair) represent PPLS estimates, the red stars and dashed vertical lines (on the right of each pair) represent PLS estimates. The dots and stars are the average loading values across 1000 simulation replications; the width of the dashed lines are twice the standard deviations. The results are for normally distributed latent variables ($t$, $e$, $f$ and $h$) and low dimensionality ($p=q=20$ variables).}
\label{fig:normlow}
\end{figure}

\begin{figure}
\begin{subfigure}[h]{.5\linewidth}
\centering
\includegraphics[width=\textwidth]{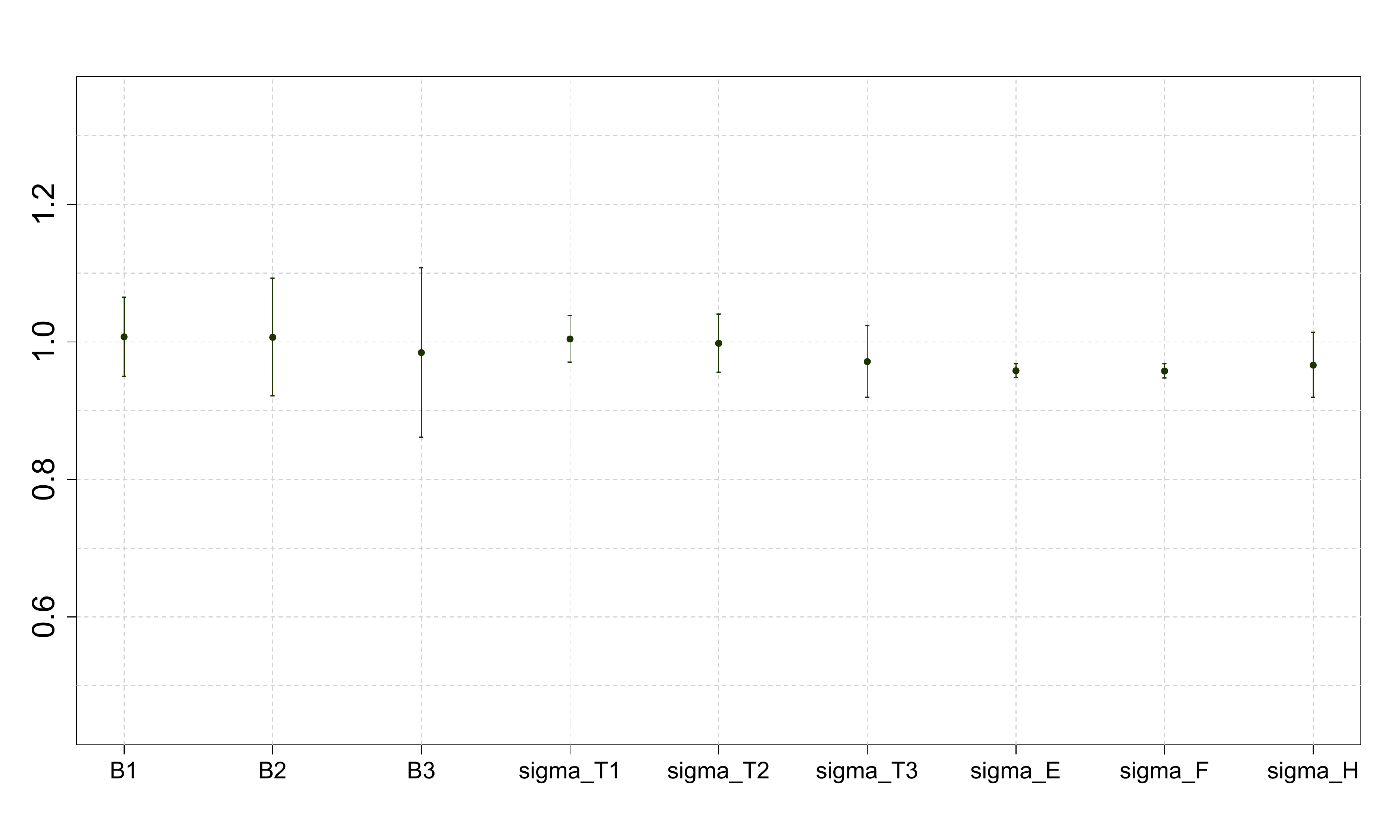}
\caption{Low noise (10\%); small sample size ($N=50$)}
\label{fig:normlowV_lowlow}
\end{subfigure}%
\begin{subfigure}[h]{.5\linewidth}
\centering
\includegraphics[width=\textwidth]{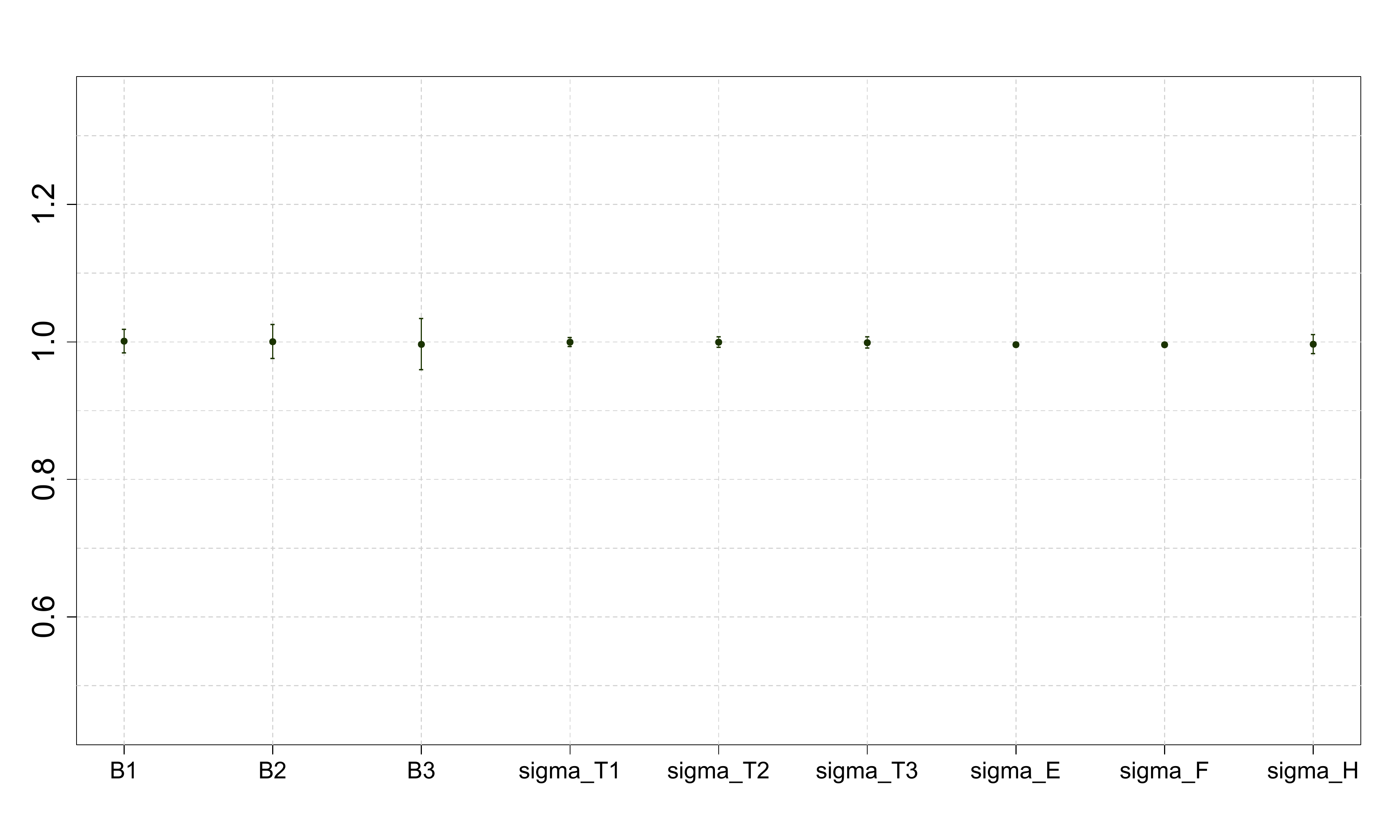}
\caption{Low noise (10\%); large sample size ($N=500$)}
\label{fig:normlowV_lowhigh}
\end{subfigure}
\\
\begin{subfigure}[h]{.5\linewidth}
\centering
\includegraphics[width=\textwidth]{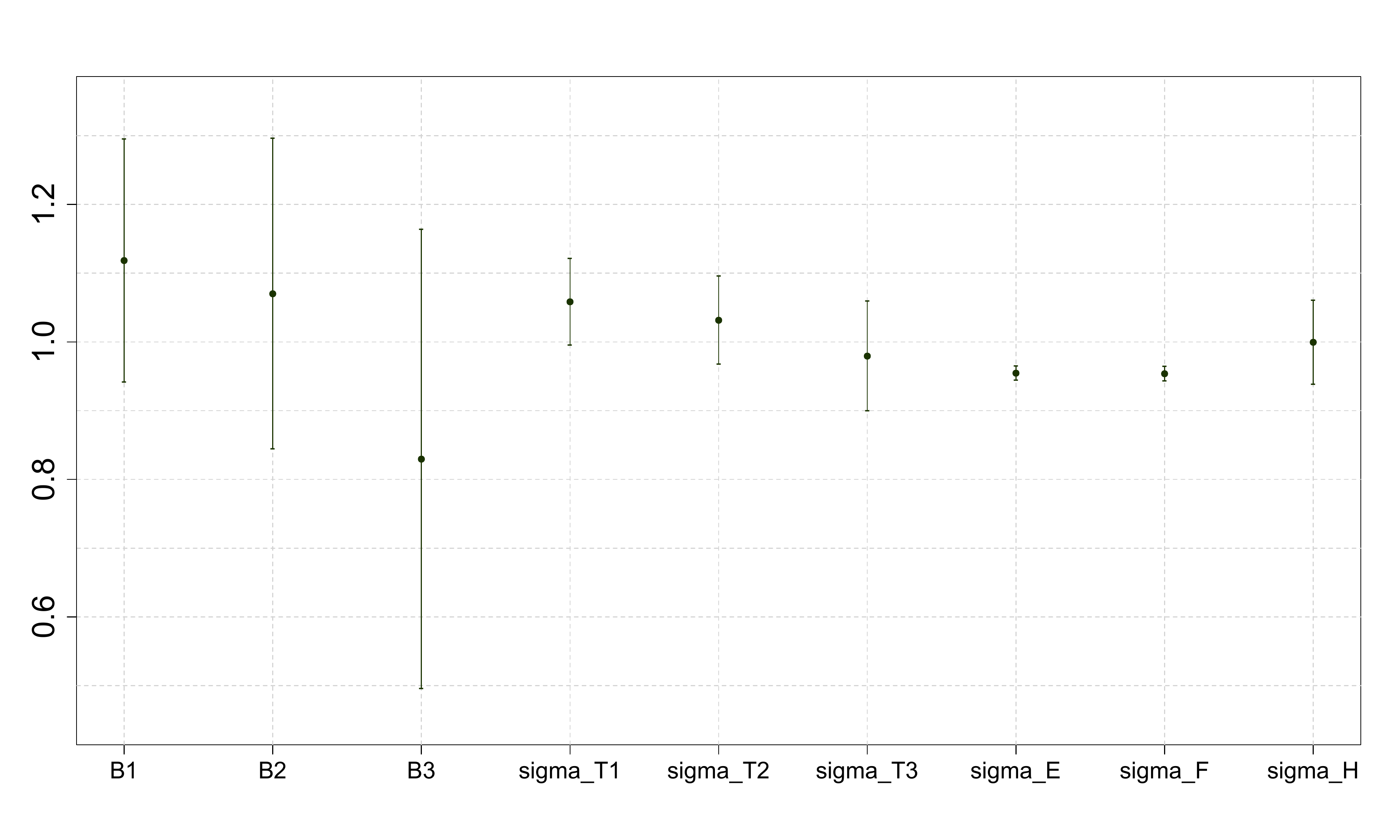}
\caption{High noise (50\%); small sample size ($N=50$)}
\label{fig:normlowV_highlow}
\end{subfigure}
\begin{subfigure}[h]{.5\linewidth}
\centering
\includegraphics[width=\textwidth]{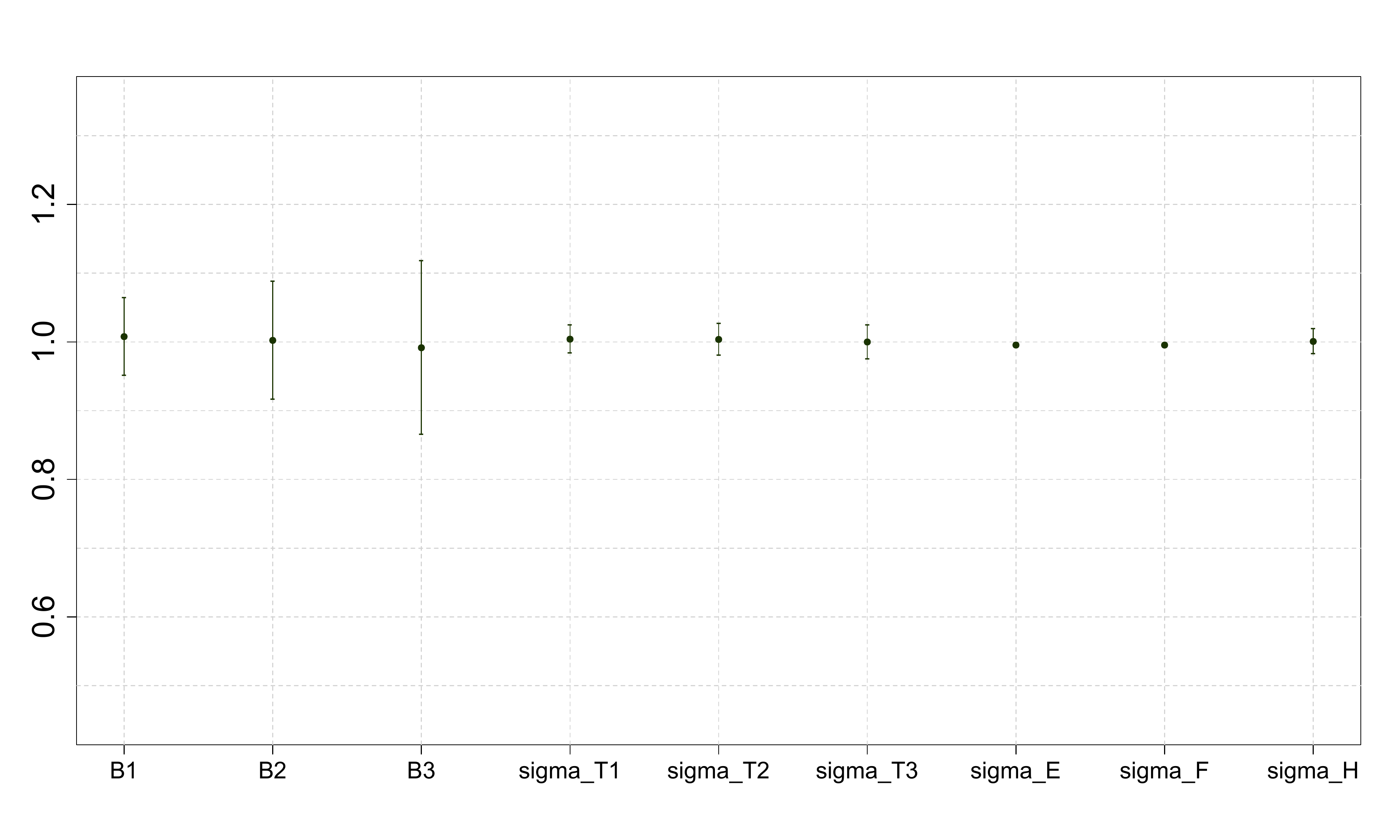}
\caption{High noise (50\%); large sample size ($N=500$)}
\label{fig:normlowV_highhigh}
\end{subfigure}
\caption{\changed{{True and estimated variance parameters $B$, $\Sigma_t$, $\sigma_e$, $\sigma_f$ and $\sigma_h$ over 1000 simulation replications.}} The dots are the average values across 1000 simulation replications; the width of the dashed lines are twice the standard deviations. The results are for normally distributed latent variables ($t$, $e$, $f$ and $h$) and low dimensionality ($p=q=20$ variables).}
\label{fig:normlowV}
\end{figure}

\begin{figure}
\centering
\includegraphics[width=\textwidth]{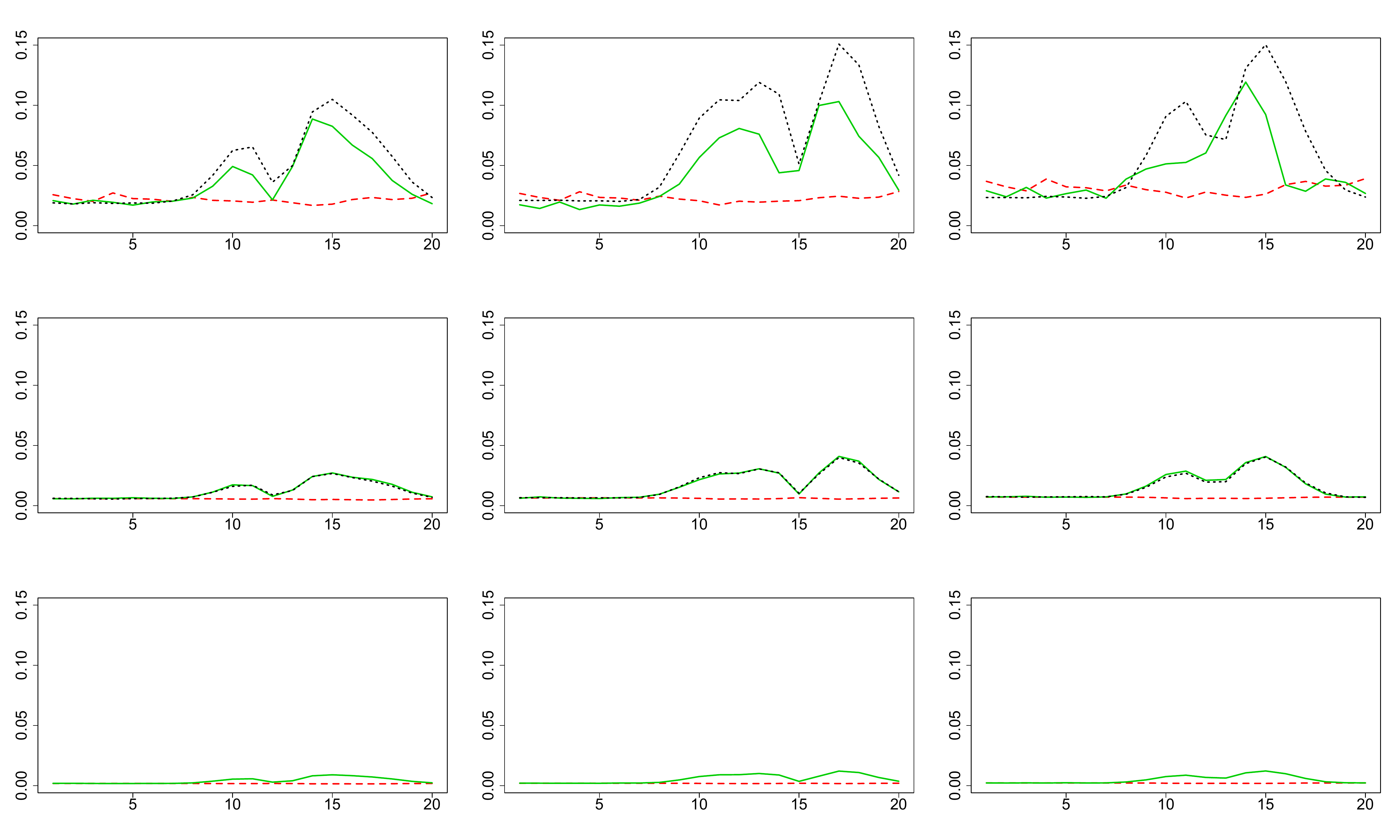}
\caption{\changed{{Standard errors of the $W$ loading elements per component.} Bootstrap standard errors (solid green line), asymptotic standard errors (dashed red line) and simulation-based standard errors (dotted black line) are plotted for the loading estimates in each component. Plots for the three sample sizes (small $N=50$, large $N=500$, `extra large' $N=5000$) are shown along the rows. The three loading components ($W_1$, $W_3$ and $W_3$) are plotted column wise. The last row does not include simulation-based standard errors. The data are generated from a normal distribution with $p=q=20$ variables and low noise level ($\alpha_n=0.1$).}}
\label{fig:W_SE_nlow}
\end{figure}

\section{Data analysis}
\label{sec:data}
\changed{
To illustrate the Probabilistic Partial Least Squares model, we apply it to IgG glycan datasets. Glycans, in particular IgG glycans, play an important role in the innate immune system, as well as in cell signaling. IgG2 glycans are less abundant than IgG1 glycans and more difficult to measure. Therefore, by using the relationships between IgG1 and IgG2 glycans, the characteristics of IgG2 can be better estimated. Hence, we will use IgG1 as $X$ matrix, and IgG2 as $Y$ matrix.
}

\changed{
In total, 40 IgG glycans were measured, of which $p=20$ are of subclass IgG1 and $q=20$ are of subclass IgG2. These data were measured in two cohorts (CROATIA\_{}Korcula with 951 participants and CROATIA\_{}Vis with 796 participants) \cite{Lauc2013}. The data matrices containing IgG1 and IgG2 glycan variables are denoted by $X_m$ and $Y_m$, with $m\in\{1,2\}$, where $m=1$ corresponds to CROATIA\_{}Korcula and $m=2$ corresponds to CROATIA\_{}Vis. 
We apply PPLS to IgG1 and IgG2 glycans in each cohort \textit{separately} and compare results. 
}

In Figure \ref{fig:Heatmap}, heatmaps of the correlations within and between the IgG1 and IgG2 glycans are shown, from which it is clear that there are many highly positive correlations between and within IgG1 and IgG2 in each data set.
\changed{The RV coefficient \cite{Robert1976}, which generalizes the \textit{squared} Pearson correlation coefficient to two matrices, was about 0.60 and 0.45 for CROATIA\_{}Korcula and CROATIA\_{}Vis cohorts respectively.}

To determine the number of latent variables to use, we considered the total amount of variance explained by the latent space relative to the total amount of variation in the data: $||T_m||/||X_m||$ and $||U_m||/||Y_m||$ for $m\in\{1,2\}$. By using four components, at least 89\% of the total variation in each of the matrices $X_1$, $X_2$, $Y_1$ and $Y_2$ was accounted for.

\changed{
For both cohorts, we fitted the PPLS models using $r=4$ latent components. The amount of overlap in each cohort, estimated by $\tr \hat{\Sigma}_{x,y} / \tr \hat{\Sigma}_y $ given in \eqref{eq:covmatrix}, was 58\% and 46\% for CROATIA\_{}Korcula and CROATIA\_{}Vis cohorts, respectively. The PPLS loadings are inspected to identify which IgG glycans contribute most to this overlap.}
The estimated IgG1 loadings $w_{j,k}$, $j\in\{1,\ldots,p\}$ and $k\in\{1,\ldots,4\}$, for both cohorts and both subclasses are depicted in Figure \ref{fig:Loadings}. The first joint component is proportional to the average glycan, as all glycans get about the same loading value. The second joint component involves especially G0 and G2 glycan subtypes, in which they are negatively correlated. Inspection of the loading values for the third component shows contibutions of fucosylated and non-fucosylated glycan subtypes. In the fourth component a pattern of positive and negative loading values is visible regarding the presence and absence of bisecting GlcNAc, respectively. The large loading value for G1NS is remarkable. The same conclusions hold for IgG2, as the estimated loading values were very similar. It is interesting to note that the observed patterns within components potentially reflect enzymatic synthesis where monosaccharides are added to glycan substrates \cite{Taniguchi2014}. Additionally, similar patterns are seen reflecting the inflammatory characteristics of glycans in aging and several different diseases \cite{Lauc2016}. Finally, the observed loading patterns were strikingly similar for both cohorts.



\begin{figure}
\begin{subfigure}[h]{\linewidth}
\centering
\includegraphics[width=0.75\textwidth,keepaspectratio]{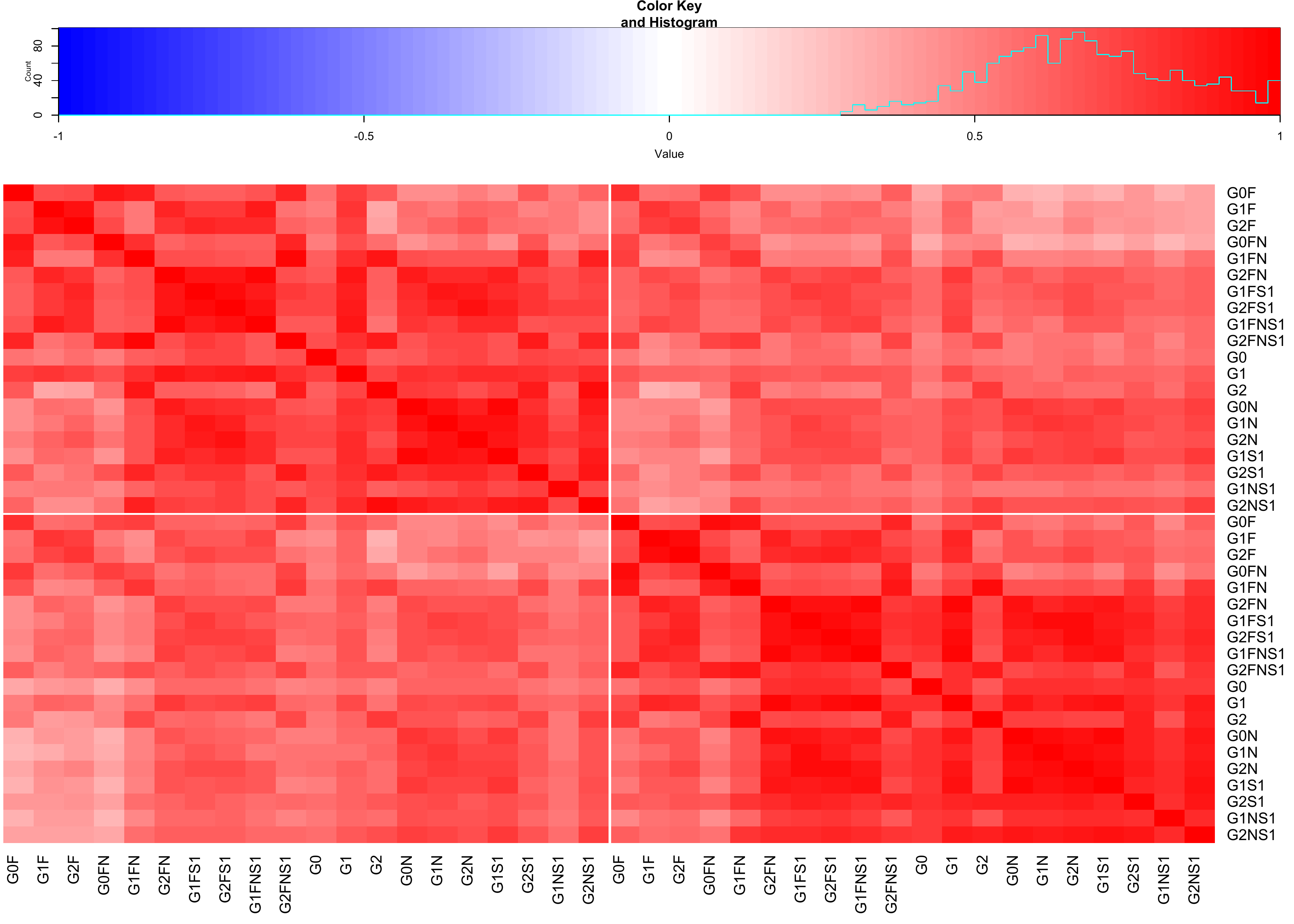}
\caption{The first (CROATIA\_{}Korcula) cohort}
\label{fig:Heatmap_Kor}
\end{subfigure}\\
\begin{subfigure}[h]{\linewidth}
\centering
\includegraphics[width=0.75\textwidth,keepaspectratio]{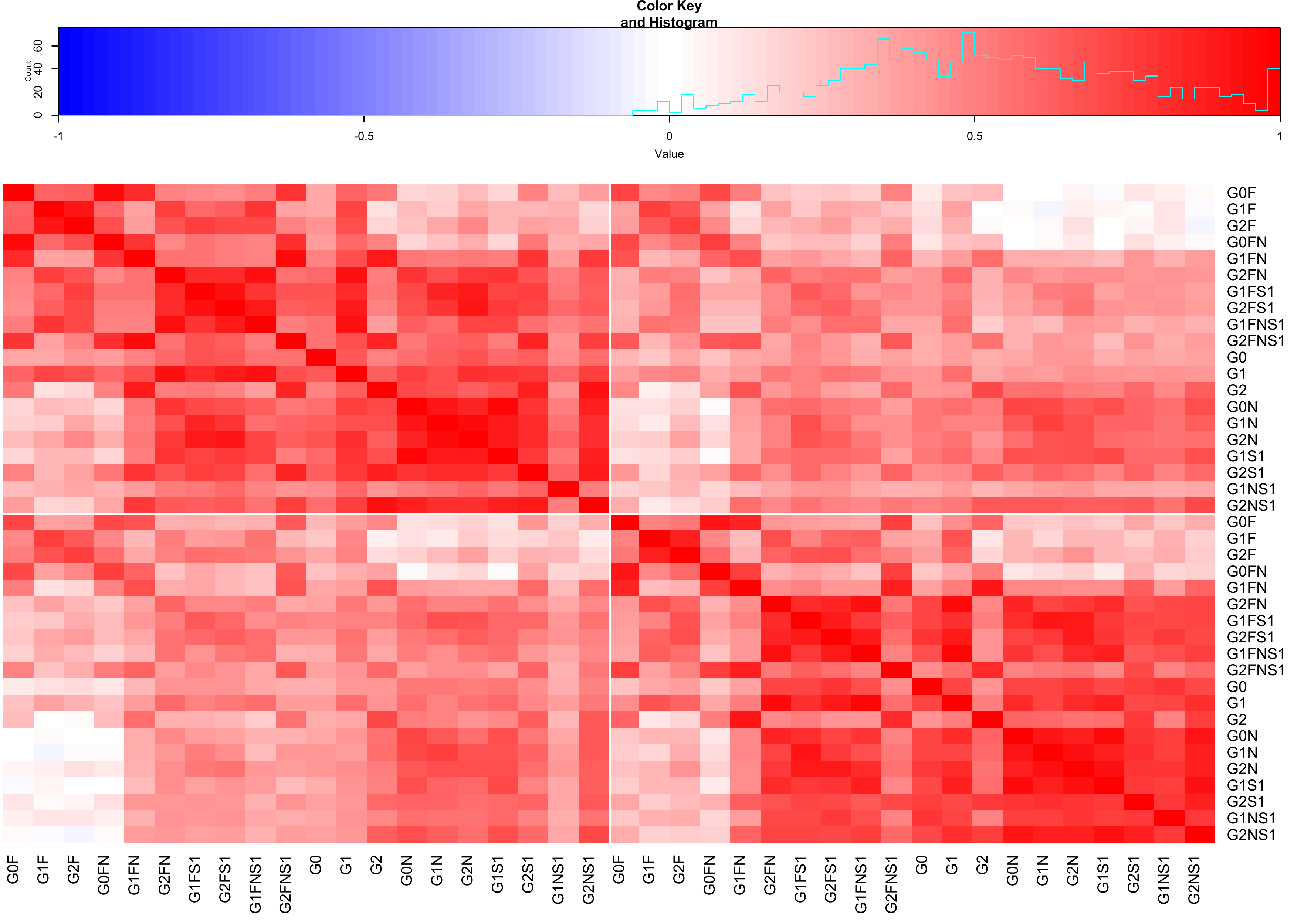}
\caption{The second (CROATIA\_{}Vis) cohort}
\label{fig:Heatmap_Vis}
\end{subfigure}
\caption{{Heatmaps of the correlations between IgG1 and IgG2 glycans.} In left the correlations of the CROATIA\_{}Korcula cohort is shown. In right the CROATIA\_{}Vis cohort is shown. The upper-left and lower-right block are the within subclass correlations, the off-diagonal block contains the correlations between IgG1 and IgG2 glycans. In both cohorts the glycans exhibit high positive correlations, especially between glycans within the IgG1 and IgG2 subclasses.}
\label{fig:Heatmap}
\end{figure}

\begin{figure}
\begin{subfigure}[h]{\linewidth}
\centering
\includegraphics[keepaspectratio,height=0.43\textheight]{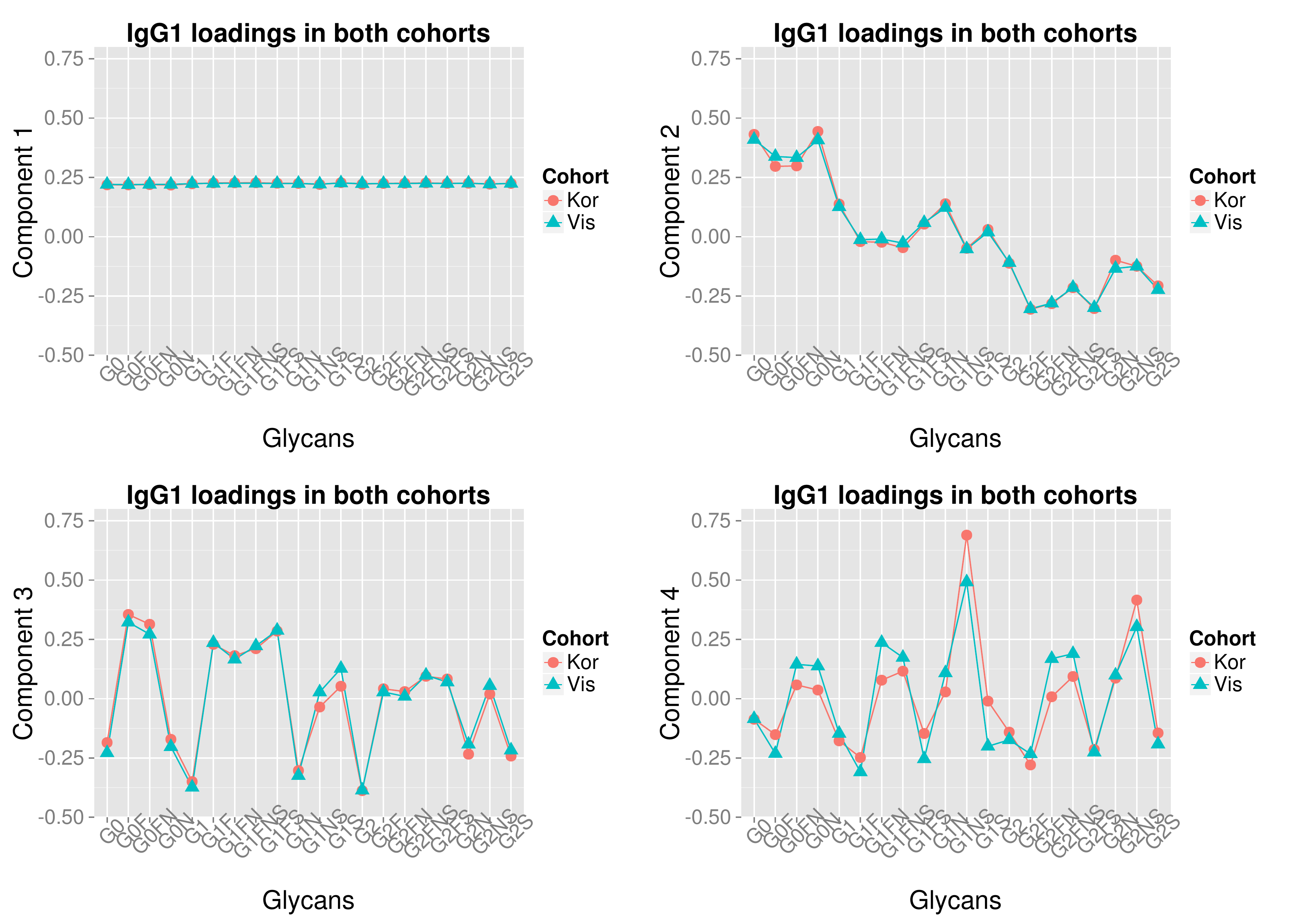}
\caption{The IgG1 glycan loadings ($W$) for both cohorts.}
\label{fig:Loadings_W}
\end{subfigure}
\\
\begin{subfigure}[h]{\linewidth}
\centering
\includegraphics[keepaspectratio,height=0.43\textheight]{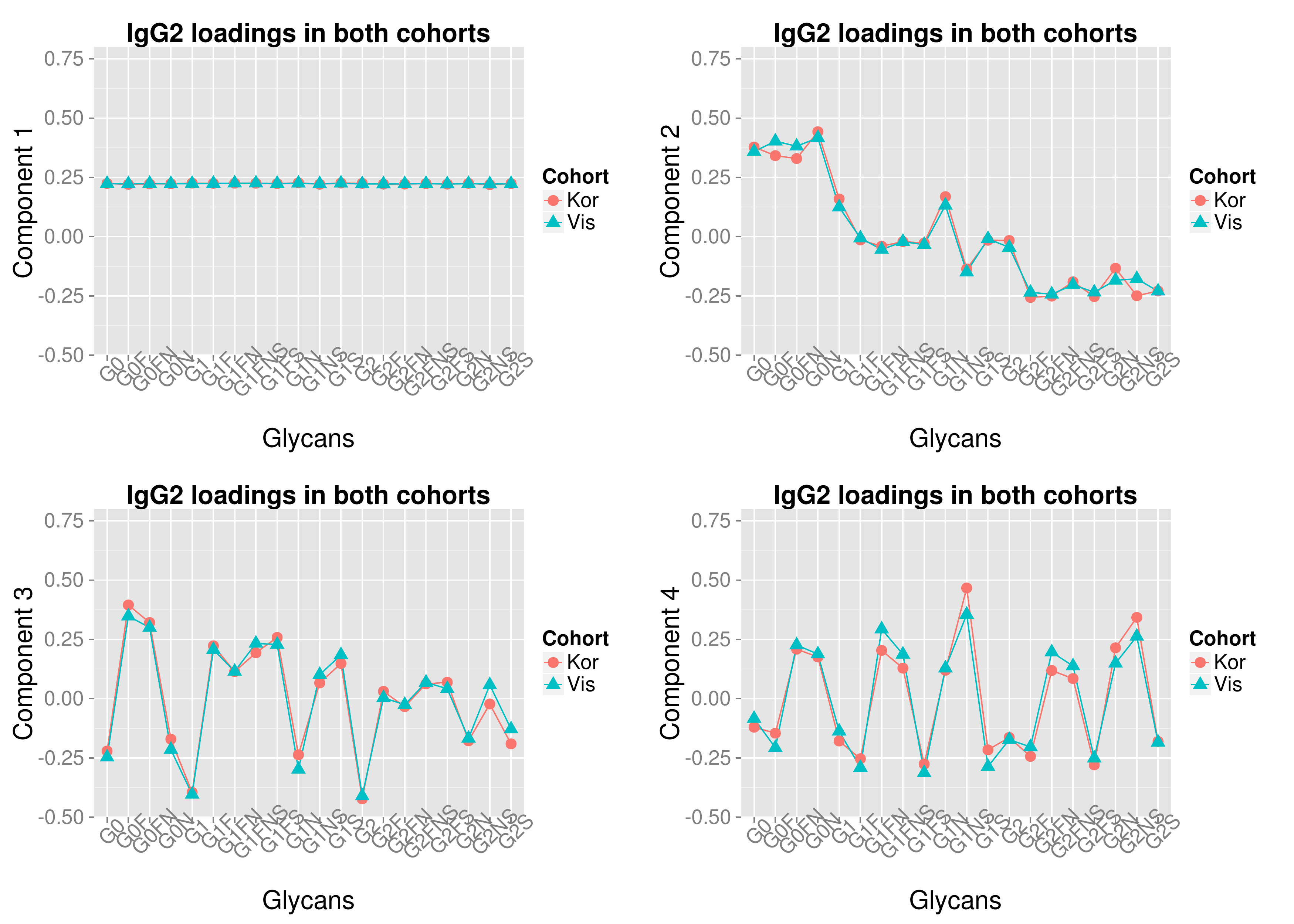}
\caption{The IgG2 glycan loadings ($C$) for both cohorts.}
\label{fig:Loadings_C}
\end{subfigure}
\caption{{Loadings per component for both cohorts.} In the top four plots loading values of IgG1 glycans ($W$) are plotted per glycan. The red dots connected by red lines are for the CROATIA\_{}Korcula loadings. The four loading vectors are plotted left-to-right and top-to-bottom. The blue triangles and lines are for the CROATIA\_{}Vis cohorts. In the bottom four plots the IgG2 glycan loading values are plotted in the same order and style.}
\label{fig:Loadings}
\end{figure}


\section{Discussion}
We proposed PPLS to model the covariance matrix of two data sets. 
Maximum likelihood estimators for the model parameters were derived by solving a constrained optimization problem, and the parameter loadings were shown to be identifiable up to sign. This property ensures that PPLS estimates are comparable across several studies.

Our simulation study showed that the PPLS estimators had good performance and lower bias compared to PLS. Most notably, the performance of PPLS was robust to misspecification of the distribution of the variables. A smaller sample size and high noise level had a negative effect on the accuracy of the estimates, but large loading values were still non-zero. \changed{Also, compared to Probabilistic CCA estimates, the PPLS estimates were less biased and more efficient. For high-dimensional data, PCCA estimates have larger bias and higher variance. This is likely to be caused by the unstable inverse sample covariance matrix calculated when using PCCA. Moreover, if the number of variables is larger than the sample size, PCCA estimates cannot be obtained. Therefore, especially in omics data analysis, PPLS provides more robust findings.
}

As an illustration of the PPLS model, we analyzed IgG glycomics data from two cohorts. The high correlations in the data (Figure \ref{fig:Heatmap}) and the use of multiple cohorts illustrate the applicability of PPLS to facilitate combination of results derived from different experimental settings. We found that the estimated loading values were almost identical across the two cohorts (Figure \ref{fig:Loadings}). 

When multiple cohorts are available, a meta-analysis on the parameter estimates can be performed. In ordinary regression models, this has been addressed for both low-dimensional \cite{DerSimonian1986} and high-dimensional \cite{He2016} design matrices. 
\changed{When there is no access to all data, PPLS estimates can be combined by using standard meta-analysis approaches \cite{DerSimonian1986}. Such an approach requires that the PPLS parameter estimates are identifiable and asymptotically normally distributed.
For the PLS framework, several approaches to combine estimates across cohorts were developed when there is access to all data. A group-PLS approach was considered \cite{Li2013} to incorporate several groups of subjects in the model. The authors showed that under certain assumptions this approach provided better predictions than a model without group effects. However their model is not identifiable and requires $N > p$.} Another method is based on weighted least squares to combine data from different studies with potentially different covariates \cite{Huang2005}. An alternative method, when access to data is possible, is to estimate joint parts between the data sets and the studies simultaneously. This yields a joint space with variables that have high loading values in most studies. For example, in \cite{VanDeun2009}, a non-probabilistic approach is pursued in a least squares estimation method using PCA. Performing data integration across studies, while taking into account uncertainties within each study, is one of our topics for future research, and will lead to more powerful analysis of the IgG glycans across cohorts.

To assess the statistical significance of loadings, the probabilistic framework provides alternative approaches to jackknifing and bootstrapping \cite{Wehrens1997}. \changed{The observed Fisher information matrix can be used to estimate standard errors for individual loading parameters. For small sample sizes, bootstrap approaches appears to better reflect the uncertainty of the parameters. For large enough sample sizes, the asymptotic standard errors are close to the simulation-based standard errors. Typically, in epidemiological studies, the sample size is large enough to use asymptotic standard errors.
}

In this paper we ignored the fact that different biological `omics' measurements have different error structures. 
An extension of Partial Least Squares was proposed to correct for systematic variation (variation induced by latent variables uncorrelated to the other data set) in the data sets, named Two-Way Orthogonal PLS (O2PLS) \cite{Bouhaddani2016,Trygg2003}. Such an extension can be pursued for PPLS by adding for both $X$ and $Y$ in \eqref{eq:PPLS} a set of independent latent variables multiplied by their loading parameters. We are currently working on exploring the possibilities of a Probabilistic O2PLS for heterogeneous data sets.

\section*{Acknowledgments}
The authors would like to thank the Editor-in-Chief, the Associate Editor and the referees for their valuable comments and suggestions.
This work has been supported by the European Union’s Seventh Framework Programme (FP7-Health-F5-2012) under grant agreement number 305280 (MIMOmics). 
The CROATIA\_{}Vis and CROATIA\_{}Korcula studies were funded by grants from the Medical Research Council (UK), European Commission Framework 6 project EUROSPAN (Contract No. LSHG-CT-2006-018947), FP7
contract BBMRI-LPC (grant No. 313010), Croatian Science Foundation (grant 8875) and the Republic of Croatia
Ministry of Science, Education and Sports (216-1080315-0302). We would like to acknowledge the staff of several institutions in Croatia that supported the field work, including but not limited to The University of Split and Zagreb Medical Schools, Institute for Anthropological Research in Zagreb and the Croatian Institute for Public Health. Glycome analysis was supported by the European Commission HighGlycan (contract No. 278535),
MIMOmics (contract No. 305280), HTP-GlycoMet (contract No. 324400), IntegraLife (contract No. 315997).
\changed{The IgG glycan data are available upon request.}

\begin{appendices}

\section{Variances and covariances}
\label{ap:theoCov}
The covariance matrix of $(x,y)$ is given in \eqref{eq:covmatrix}. First note that $\Var(u) = \Var(tB+h) = B^2 \Sigma_t + \sigma_h^2 I_r$, then compute
\begin{equation*}
\begin{split}
\Var(x) & = \Var (tW^\top + e ) = W\Var(t)W^\top + \Var(e) = W\Sigma_t W^\top + \sigma_e^2 I_p, \\
\Var(y) & = \Var (UC^\top + f ) = C\Var(u)C^\top + \Var(f) = C (B^2 \Sigma_t + \sigma_h^2 I_r) C^\top + \sigma_f^2 I_q,  \\
\Cov(x,y) & = \Cov (tW^\top + e,uC^\top+f ) = W\Cov(t,u)C^\top = W\Cov(t,tB)C^\top = W B \Sigma_t C^\top.
\end{split}
\end{equation*}
The covariances between the observed and latent variables are as follows
\begin{equation*}
\begin{split}
\Cov(x,t) & = \Cov (tW^\top + e,t ) = W\Var(t) = W\Sigma_t, \\
\Cov(x,u) & = \Cov (tW^\top + e,tB + h ) = W\Var(t)B = W\Sigma_t B, \\
\Cov(y,t) & = \Cov (uC^\top + f,t ) = C\Cov(tB+h,t) = C\Sigma_tB, \\
\Cov(y,u) & = \Cov (uC^\top + f,u ) = C\Cov(tB+h,tB+h) = C (\Sigma_tB^2 + \sigma_h^2 I_r  ).
\end{split}
\end{equation*}
See, e.g., \cite{Seber2003} for more details.

\section{Identifiability of PPLS}
\label{ap:identif}
For establishing identifiability of the PPLS model, we need to prove that if the distribution of $(x,y)$ is given, there is only one corresponding set of parameters yielding this distribution. Since $(x,y)$ follows a zero mean normal distribution, identifiability is equivalent to 
\begin{equation*}
\Sigma = \tilde{\Sigma} \quad \Leftrightarrow \quad \theta = \tilde{\theta},
\end{equation*}
where $\Sigma,\tilde{\Sigma}$ is defined, through $\theta,\tilde{\theta}$, in \eqref{eq:covmatrix}.
The following lemma will be very useful in establishing identifiability.

\begin{lemma}{{\rm (Singular Value Decomposition).}}
\label{lemma:SVD}
Let $W, \tilde{W}$ be $p\times r$ and $C, \tilde{C}$ be $q\times r$, all orthogonal matrices. Let $D, \tilde{D}$ be $r \times r$ diagonal with $r$ distinct positive elements on the diagonal.
Then $W D C^\top = \tilde{W} \tilde{D} \tilde{C}^\top$ (B.1) implies $W = \tilde{W} \Delta$, $C = \tilde{C} \Delta$ for some diagonal matrix $\Delta$ of size $r \times r$ with on the diagonal elements $\delta_i \in \{-1,1\}$ and $D = \tilde{D}$.
\end{lemma}
\begin{proof}
Let $A_1 = W D C^\top$ and $A_2 = \tilde{W} \tilde{D} \tilde{C}^\top$. Consider $A_iA_i^\top$ and $A_i^\top A_i$, $i\in\{1,2\}$. 
The assertion (B.1) then implies the following.
\begin{equation*}
A_1 A_1^\top = W D^2 W^\top = \tilde{W} \tilde{D}^2 \tilde{W}^\top = A_2 A_2^\top;\quad
A_1^\top A_1 = C D^2 C^\top = \tilde{C} \tilde{D}^2 \tilde{C}^\top = A_2^\top A_2.
\end{equation*}
Note that both $W D^2 W^\top$ and $\tilde{W} \tilde{D}^2 \tilde{W}^\top$ are eigenvalue decompositions, as $D^2$ and $\tilde{D}^2$ are diagonal and $W,\tilde{W}$ and $C, \tilde{C}$ are orthogonal. The spectral theorem for matrices \cite{Eaton1983} then implies that whenever the elements in $D^2,\tilde{D}^2$ are distinct, the corresponding columns in $W,\tilde{W}$ and $C, \tilde{C}$ are equal up to multiplication with the same sign. We thus get $W = \tilde{W} \Delta$, $C = \tilde{C} \Delta$ and $D = \tilde{D}$.
\end{proof}
Using this lemma, we show identifiability of the off-diagonal block part of the covariance matrix as given in Eq.~\eqref{eq:covmatrix}.
\begin{lemma}
\label{th:XY}
If for matrices $W,\tilde{W}$, $C,\tilde{C}$ and diagonal $B,\tilde{B}$ and $\Sigma_{t},\tilde{\Sigma}_{t}$, given as in the PPLS model, $W \Sigma_{t}B C^\top = \tilde{W} \tilde{\Sigma}_{t}\tilde{B} \tilde{C}^\top$, then $W = \tilde{W} \Delta$, $C = \tilde{C} \Delta$ and $\Sigma_{t}B = \tilde{\Sigma}_{t}\tilde{B}$.
\end{lemma}
\begin{proof}
Applying Lemma \ref{lemma:SVD} with $D = \Sigma_{t} B$ and $\tilde{D} = \tilde{\Sigma}_{t}\tilde{B}$ gives the desired result, since $\Sigma_{t} B$ and $\tilde{\Sigma}_{t}\tilde{B}$ are diagonal matrices with distinct ordered elements.
\end{proof}
Given $\Sigma_{x,y}$ we can identify $W$ and $C$ up to sign and the product $\Sigma_{t}B$. We now show that in particular also the individual parameters $\Sigma_{t}$ and $B$ are identified from the upper diagonal block matrix $\Sigma_x$.
\begin{lemma}
\label{th:x}
If for matrix $W$, diagonal matrices $\Sigma_{t}$ and $\tilde{\Sigma}_{t}$ and positive numbers $\sigma_{e}^2,\tilde{\sigma}_e^2$, given as in the PPLS model,
$W \Sigma_{t} W^\top + \sigma_{e}^2 I_p= W \tilde{\Sigma}_{t} W^\top + \tilde{\sigma}_{e}^2 I_p$ (B.2), then $\sigma_{e} = \tilde{\sigma}_{e}$ and $\Sigma_{t} = \tilde{\Sigma}_{t}$.
\end{lemma}
\begin{proof}
Suppose (B.2) holds. Since $r < p$ and $p > 1$, one can find a unit vector $w_\perp$ such that $W^\top w_\perp = 0$. Multiplying with such vector yields
$\sigma_{e}^2 w_\perp = \tilde{\sigma}_{e}^2 w_\perp$.
Multiplying again with $w^\top_\perp$ yields
$\sigma_{e}^2 = \tilde{\sigma}_{e}^2$.
It follows that we can identify $\sigma_e^2$.
We can now reduce (B.2) to
$W \Sigma_{t} W^\top = W \tilde{\Sigma}_{t} W^\top$.
Pre-multiplying with $W^\top$ and post-multiplying with $W$ on both sides yields
$\Sigma_{t} = \tilde{\Sigma}_{t}$.
\end{proof} 
We have seen in Theorem \ref{th:XY} that we can identify $\Sigma_{t}B$. Since we identified $\Sigma_t$ we get identifiability of $B$. The remaining parameters $\sigma_h^2$ and $\sigma_f^2$ are now shown to be identified using the lower block diagonal $\Sigma_y$.
\begin{lemma}
\label{th:y}
If for matrices $C$, $B$, $\Sigma_{t}$, $\sigma_{f}^2, \tilde{\sigma}_f^2$ and $\sigma_{h}^2, \tilde{\sigma}_{h}^2$, given as in the PPLS model, the assertion $\Sigma_{y} = \tilde{\Sigma}_{y} $ holds, i.e.,  
\begin{equation*}
C (B^2 \Sigma_t + \sigma_{h}^2 I_r) C^\top + \sigma_{f}^2 I_q = C (B^2 \Sigma_t + \tilde{\sigma}_{h}^2 I_r) C^\top + \tilde{\sigma}_{f}^2 I_q,
\end{equation*}
then $\sigma_{f}^2 = \tilde{\sigma}_f^2$ and $\sigma_{h}^2 = \tilde{\sigma}_{h}^2$.
\end{lemma}
\begin{proof}
In Theorem \ref{th:x} take $W$ equal to $C$, $\sigma_e^2$ equal to $\sigma_f^2$, $\tilde{\sigma}_e^2$ equal to $\tilde{\sigma}_f^2$, and the diagonal covariance matrices $\Sigma_t$ and $\tilde{\Sigma}_t$ equal to $\Sigma_t B^2 + \sigma_h^2 I_p$ and $\Sigma_t B^2 + \tilde{\sigma}_h^2 I_p$. We find that we can identify $\Sigma_{t}B^2 + \sigma_h^2$ and $\sigma_f^2$. Since we already identified $\Sigma_{t}$ and $B$, we have also identifiability of $\sigma_h^2$.
\end{proof}
We conclude with the proof of Theorem \ref{th:identif}.
\begin{proof}
Suppose $\Sigma = \tilde{\Sigma}$. This is true if and only if 
\begin{equation}
\tag{B.1}
\Sigma_{x,y} = \tilde{\Sigma}_{x,y}, \quad
\Sigma_{x} = \tilde{\Sigma}_{x}, \quad
\Sigma_{y} = \tilde{\Sigma}_{y}. 
\label{eq:blocks}
\end{equation}
Applying Lemma \ref{th:XY} to the first equation, we identify $W$ and $C$ up to sign. 
Considering Lemma \ref{th:x} together with Lemma \ref{th:XY}, the second equation implies identifiability of $\Sigma_t$, $B$ and $\sigma_e$.
The three Lemmas \ref{th:XY}, \ref{th:x} and \ref{th:y} together with the last equation imply identifiability of $\sigma_h$ and $\sigma_f$.

\end{proof}

\section{An Expectation-Maximization algorithm for PPLS}
\label{ap:EM}
To obtain parameter estimates in the PPLS model, maximum likelihood is used. The EM algorithm is an iterative procedure for maximizing the observed log-likelihood \eqref{eq:loglike_obs} and consists of an Expectation step and a Maximization step. The following Lemma is convenient to make the expectation step explicit.
\begin{lemma}\label{lemma:schur}
    Let the pair $(z,x)$ be jointly multivariate normal row vectors with {zero} mean and covariance matrix
\begin{equation*}
	\begin{pmatrix} \Sigma_z & \Sigma_{z,x} \\ \Sigma_{x,z} & \Sigma_x \end{pmatrix}.
\end{equation*}
Then $z|x$ is normally distributed with conditional mean $\mathrm{E}\left( z|x \right) = x\,\Sigma_x^{-1}\,\Sigma_{x,z}$,
and conditional covariance matrix $\Var \left( z|x \right) = \Sigma_z - \Sigma_{z,x}\,\Sigma_x^{-1}\,\Sigma_{x,z}$.
Secondly, if $z=(t,u)$, $\Cov(t,x)=\Sigma_{t,x}$ and $\Cov(x,u)=\Sigma_{x,u}$, then the conditional covariance between $t$ and $u$ is $\Cov(t,u|x) = \Cov(t,u) - \Sigma_{t,x}\,\Sigma_x^{-1}\,\Sigma_{x,u}$.
\end{lemma}
\begin{proof}
The proof for the first part of the Lemma is found in \cite{Seber2003}. The second part follows from the off diagonal blocks of $\Var(z|x)$.
\end{proof}

\paragraph{Expectation}
The conditional first moments can be obtained by applying Lemma \ref{lemma:schur} while substituting $t$ or $u$ for $z$ and $(x,y)$ for $x$.
\begin{equation*}
\mu_t = \mathrm{E}\left( t | x,y,\theta \right) = (x,y)\,\Sigma^{-1}\,\Cov\{(x,y),t\}, \quad
\mu_u = \mathrm{E}\left( u | x,y,\theta \right) = (x,y)\,\Sigma^{-1}\,\Cov\{(x,y),u\}.
\end{equation*}
The same substitution can be made for the conditional second moments. Using $\mathrm{E}(a^\top b|z) = \Cov(a,b|z) + \mathrm{E}(a|z)^\top \mathrm{E}(b|z)$, we get
\begin{equation*}
\begin{split}
C_{TT} = \mathrm{E}( t^\top t | x,y,\theta ) & = I_r - \Cov\{t,(x,y)\}\,\Sigma^{-1}\,\Cov\{(x,y),t\} + \Cov\{t,(x,y)\}\,\Sigma^{-1} S \Sigma^{-1}\,\Cov\{(x,y),t\}, \\
C_{UU} = \mathrm{E}( u^\top u | x,y,\theta ) & = I_r - \Cov\{u,(x,y)\}\,\Sigma^{-1}\,\Cov\{(x,y),u\} + \Cov\{u,(x,y)\}\,\Sigma^{-1} S \Sigma^{-1}\,\Cov\{(x,y),u\},
\end{split}
\end{equation*}
where $S$ is the biased sample covariance matrix of $(x,y)$.
The conditional cross term equals 
\begin{equation*}
\begin{split}
C_{UT} = \mathrm{E}( u^\top t | x,y,\theta ) & = \Sigma_t B - \Cov\{u,(x,y)\}\,\Sigma^{-1}\,\Cov\{(x,y),t\} + \Cov\{u,(x,y)\}\,\Sigma^{-1} S \Sigma^{-1}\,\Cov\{(x,y),t\}
\end{split}
\end{equation*}
The covariances are given by
\begin{equation*}
\Cov\{(x,y),t\} = \begin{pmatrix} W\Sigma_t \\ C \Sigma_t B \end{pmatrix}, \quad \Cov\{(x,y),u\} = \begin{pmatrix} W \Sigma_t B \\ C (\Sigma_t B + \sigma_h^2 I_r) \end{pmatrix}.
\end{equation*}
Although the the conditional expectations involve random variables and parameters, in the maximization step the calculated quantities are considered fixed and known.

\paragraph{Maximization}
The maximization step involves maximizing the complete-data likelihood \eqref{eq:decomposition}, we have seen that it can be decomposed in distinct factors. This allows optimization of the expected complete data likelihood to be split into four sub-maximizations, given by the individual factors and their respective parameters in the following annotated product:
\begin{equation*}
\undertext{f(x|t)}{W,\sigma_e}\,\undertext{f(y|u)}{C,\sigma_f}\,\undertext{f(u|t)}{B,\sigma_h}\,\undertext{f(t)}{\Sigma_t}
\end{equation*}
Moreover, it will become apparent that each parameter within each component can be decoupled, yielding a separate maximization per component per parameter.
We focus on the part of $f(x|t)$ that depends on $W$, which is given~by
\begin{equation*}
\begin{split}
\mathrm{E}\left\{ \ln f(X|T) | X,Y \right\} & = -\mathrm{E}( ||X - TW^\top||^2 | X,Y ) + \mathrm{const.} \\
& = \tr( - X^\top X + 2 X^\top \mu_t W^\top - WC_{TT}W^\top ) + \mathrm{const.}
\end{split}
\end{equation*}
To take into account the constraints on $W$, namely $W^\top W = I_r$, we introduce a matrix of Lagrange multipliers $\Lambda$. We get as objective function
\begin{equation*}
\tr( - X^\top X + 2 X^\top \mu_t W^\top - WC_{TT}W^\top ) - \tr\{( W^\top W - I_r ) \Lambda\}.
\end{equation*}
Differentiating with respect to $W$ yields
$2 X^\top \mu_t - 2 W C_{TT} - 2 W \Lambda = 2 W \left(C_{TT} + \Lambda \right) - 2 X^\top \mu_t$.
One may choose $\Lambda$ so that $C_{TT} + \Lambda$ is invertible. In a maximum $W$ then satisfies
$W = X^\top \mu_t \left(C_{TT} + \Lambda \right)^{-1}$.
We want to find a $\Lambda$ such that the constraint holds, i.e., 
\begin{equation*}
I_r = W^\top W = \{(C_{TT} + \Lambda )^{-1}\}^\top \mu_t^\top X X^\top \mu_t \left(C_{TT} + \Lambda \right)^{-1}, \quad
\mu_t^\top X X^\top \mu_t = \left(C_{TT} + \Lambda \right)^\top \left(C_{TT} + \Lambda \right)
\end{equation*}
The last identity can be recognized as a Cholesky or Eigenvalue decomposition.
\begin{equation*}
\mu_t^\top X X^\top \mu_t = \left(C_{TT} + \Lambda \right)^\top \left(C_{TT} + \Lambda \right) = L_t L_t^\top
\end{equation*}
with $L_t$ the lower triangular matrix of a Cholesky decomposition of $\mu_t^\top X X^\top \mu_t$. Note that $L_t$ exists, since $\mu_t^\top X X^\top \mu_t$ is always positive semi-definite. 
Choosing $\Lambda = L_t^\top - C_{TT}$, we get as update $W = X^\top \mu_t (L_t^\top)^{-1}$.
Following the same reasoning, we obtain for the $f(Y|U)$ part $C = Y^\top \mu_u (L_u^\top)^{-1}$, where $L_u$ is the lower triangular matrix from the Cholesky decomposition of $\mu_u^\top Y Y^\top \mu_u$. 

The parameter $B$ involves maximizing $\ln f(U|T)$, which is given by
\begin{equation*}
- || U - TB ||^2  = - \tr( U^\top U - 2 U^\top TB + B T^\top T B ) + \mathrm{const.}
\end{equation*}
Taking the conditional expectation with respect to $(x,y)$ yields
$- \tr \, \mathrm{E}( U^\top U - 2U^\top TB + B T^\top T B | X,Y )$.
Differentiating with respect to $B$ and equating to the zero matrix yields
\begin{equation*}
B \mathrm{E}( T^\top T | X,Y ) = \mathrm{E}( U^\top T | X,Y ) \quad
B = \mathrm{E}( U^\top T | X,Y ) \{ \mathrm{E}( T^\top T | X,Y ) \}^{-1}
\end{equation*}
To incorporate the constraint that $B$ should be diagonal, we set the diagonal elements to zero, yielding
\begin{equation*}
B = \mathrm{E}( U^\top T | X,Y ) \{ \mathrm{E}( T^\top T | X,Y ) \}^{-1} \circ I_r,
\end{equation*}
with $\circ$ the element-wise (Hadamard) product operator.

For the covariance matrix of $\Sigma_t$, we consider $\ln f(T)$ which is given by
\begin{equation*}
2\ln f(T) = \mathrm{const.} - N\ln |\Sigma_T| - \tr( T^\top T \Sigma_T^{-1} ) 
= \mathrm{const.} + N \ln |\Sigma_T^{-1}| - \tr( T^\top T \Sigma_T^{-1} ).
\end{equation*}
After taking the conditional expectation of the last expression, it can be differentiated with respect to $\Sigma_t^{-1}$, which yields
\begin{equation*}
2\diffp{\Sigma_T^{-1}} \ln f(T) = N \Sigma_T - \mathrm{E}( T^\top T | x,y ) = 0, \quad
\Sigma_T = N^{-1} \mathrm{E}( T^\top T | x,y ) \circ I_r
\end{equation*}
The last Hadamart product ensures $\Sigma_t$ is diagonal.

To maximize over $\sigma_e^2$, we consider $\ln f(X|T)$ and note that $E = X - T W^\top$. Then $\ln f(X|T)$ is given by
\begin{equation*}
2\ln f(X|T) = \mathrm{const.} - Np \ln |\sigma_e^2| - \sigma_e^{-2} \tr( E^\top E ) = \mathrm{const.} + Np \ln \sigma_e^{-2} - \sigma_e^{-2} \tr( E^\top E ) 
\end{equation*}
After taking the conditional expectation of the last expression, we differentiate it with respect to $\sigma_e^{-2}$, yielding
\begin{equation*}
2\diffp{\sigma_e^{-1}} \ln f(X|T) = Np \sigma_e^2 -  \mathrm{E}( E^\top E | X,Y ) = 0, \quad
\sigma^2_e = (Np)^{-1} \mathrm{E}( E^\top E | X,Y ) 
\end{equation*}
The same derivation can be applied to $\ln f(y|u)$ and $\ln f(u|t)$ to find
\begin{equation*}
\sigma^2_f = (Nq)^{-1} \mathrm{E}( F^\top F | X,Y ), \quad
\sigma^2_h = (Nr)^{-1} \mathrm{E}( H^\top H | X,Y )
\end{equation*}

\section{Asymptotic standard errors for PPLS loadings}
\label{ap:SE}
\changed{
Using notation as in \cite{Louis1982} we define 
\begin{equation*}
\lambda(W_k) = -\frac{1}{2\sigma_e^2} \tr( X^\top X - 2 X^\top t_k w_k^\top + w_k t_k^\top t_k w_k^\top )
\end{equation*}
to be the part of the log likelihood depending on $w_k$.
We calculate the following first and second derivatives.
\begin{equation*}
S(w_k) = \nabla \lambda = \sigma_e^{-2}  ( X^\top t_k - w_k t_k^\top t_k ), \quad
B(w_k) = -\nabla^2 \lambda = \sigma_e^{-2}  (t_k^\top t_k ) I_p.
\end{equation*}
We obtain
\begin{equation*}
\begin{split}
\sigma_e^4 S(w_k)S(w_k)^\top & = X^\top t_k t_k^\top X - 2 X^\top t_k t_k^\top t_k w_k^\top + w_k t_k^\top t_k t_k^\top t_k w_k^\top, \\
\sigma_e^4 \mathrm{E} \{S(w_k)S(w_k)^\top|X,Y \} & = X^\top \mathrm{E} (t_k t_k^\top|X,Y ) X - 2 X^\top \mathrm{E} (t_k t_k^\top t_k|X,Y ) w_k^\top + w_k \mathrm{E} (t_k^\top t_k t_k^\top t_k|X,Y ) w_k^\top \\
& = \sigma_k^2 X^\top X - 2 X^\top ( \mu_k ||\mu_k ||^2_2 + 3\mu_k \sigma_k^2  ) w_k^\top + w_k ( ||\mu_k ||^4_2 + 6||\mu_k ||^2_2 \sigma_k^2  + 3 \sigma_k^4  ) w_k^\top.
\end{split}
\end{equation*}
Here $\mu_k = \mathrm{E}(t_k|X,Y)$ and $\sigma_k = \mathrm{E}(t_k^\top t_k|X,Y)$. For explicit expressions of these expectations, see Appendix \ref{ap:EM}.
For the second derivative we get $\mathrm{E}\{ B(w_k)|X,Y \} = \sigma_k^2 I_p / \sigma_e^2$.
The observed Fisher information is now
\begin{equation*}
I_{\rm obs} = \mathrm{E}\{B(w_k)|X,Y\} - \mathrm{E}\{S(w_k)S(w_k)^\top|X,Y\},
\end{equation*}
and the asymptotic covariance matrix of $w_k$ is $-I_{\rm obs}^{-1}$. The square root of the diagonal elements are the standard errors of the corresponding loading elements.
}
\end{appendices}

\end{document}